\documentclass{article}

\usepackage[english]{babel}

\usepackage[letterpaper,top=1in,bottom=1in,left=1in,right=1in,marginparwidth=1.75cm]{geometry}

\usepackage[pagewise]{lineno}%\linenumbers

\usepackage{amsmath,amsfonts,mathtools,amsthm}
\usepackage[most]{tcolorbox}
\usepackage{graphicx}
\usepackage[colorlinks=true, allcolors=blue]{hyperref}
\usepackage{cleveref}
\usepackage{algorithm}
\usepackage{algpseudocode}
\usepackage{lscape}
\usepackage{booktabs}
\usepackage{authblk}
\usepackage{geometry}
\usepackage{pdflscape}
\usepackage{makecell}
\usepackage{utfsym}

\usepackage{hyperref}
\usepackage{caption}

\usepackage[textsize=tiny]{todonotes}
\usepackage[caption=false]{subfig}

\newenvironment{@abssec}[1]{%
     \if@twocolumn
       \section*{#1}%
     \else
       \vspace{.05in}\footnotesize
       \parindent .2in
         {\upshape\bfseries #1. }\ignorespaces
     \fi}
     {\if@twocolumn\else\par\vspace{.1in}\fi}

\newenvironment{keywords}{\begin{@abssec}{\keywordsname}}{\end{@abssec}}

\newenvironment{AMS}{\begin{@abssec}{\AMSname}}{\end{@abssec}}

\definecolor{darkred}{rgb}{.7,0,0}

\definecolor{darkgreen}{rgb}{.15,.55,0}

\definecolor{darkblue}{rgb}{0,0,0.7}

\newcommand{\A}{\mathbf{A}}
\newcommand{\B}{\mathbf{B}}
\newcommand{\bC}{\mathbf{C}}
\newcommand{\D}{\mathbf{D}}

\newcommand{\J}{\mathbf{J}}

\newcommand{\R}{\mathbb{R}}

\newcommand{\Ex}{\mathbb{E}}

\newcommand{\de}{\text{d}}

\newcommand{\Tr}{\text{Tr}\,}

\newcommand{\etaref}{\eta^0} % also consider \eta_0, or instead \varrho

\newcommand\keywordsname{Key words}
\newcommand\AMSname{AMS subject classifications}

\newtheorem{theorem}{Theorem}[section]

\newtheorem{lemma}{Lemma}[section]

\newtheorem{remark}{Remark}[section]

\newtheorem{assumption}{Assumption}[section]

\usepackage{comment}

\title{\sc Transport map unadjusted Langevin algorithms: learning and discretizing perturbed samplers}

\author[1]{Benjamin J.\ Zhang\footnote{Corresponding author. Email: \texttt{bjzhang@umass.edu}}$^\dagger$}
\author[2]{Youssef M.\ Marzouk\footnote{BJZ and YMM were partially supported by the AFOSR MURI
Analysis and Synthesis of Rare Events, award FA9550-20-1-0397. }}
\author[3]{Konstantinos Spiliopoulos\footnote{KS was partially supported by NSF DMS-2107856, DMS-2311500 and Simons Foundation Award  672441.}}
\affil[1]{\small{Department of Mathematics and Statistics, University of Massachusetts Amherst}}

\affil[2]{\small{Department of Aeronautics and Astronautics, Massachusetts Institute of Technology}}
\affil[3]{\small{Department of Mathematics and Statistics, Boston University}}

\begin{document}
\maketitle
\begin{abstract}
  Langevin dynamics are widely used in sampling high-dimensional, non-Gaussian distributions whose densities are known up to a normalizing constant. In particular, there is strong interest in unadjusted Langevin algorithms (ULA), which directly discretize Langevin dynamics to estimate expectations over the target distribution.  We study the use of transport maps that approximately normalize a target distribution as a way to precondition and accelerate the convergence of Langevin dynamics.  We show that in continuous time, when a transport map is applied to Langevin dynamics, the result is a Riemannian manifold Langevin dynamics (RMLD) with metric defined by the transport map. We also show that applying a transport map to an irreversibly-perturbed ULA results in a geometry-informed irreversible perturbation (GiIrr) of the original dynamics. These connections suggest more systematic ways of learning metrics and perturbations, and also yield alternative discretizations of the RMLD described by the map,  which we study.
  Under appropriate conditions, these discretized processes can be endowed with non-asymptotic bounds describing convergence to the target distribution in 2-Wasserstein distance. Illustrative numerical results complement our theoretical claims.
 \end{abstract}


\section{Introduction}
\label{sec:intro}
\normalsize

Langevin dynamics frequently appear in sampling methods for high dimensional, non-Gaussian probability distributions. Given access to the gradient of the logarithm of the target density, many Markov chain Monte Carlo (MCMC) methods use Langevin dynamics as a core ingredient. Particular recent interest is in unadjusted Langevin algorithms (ULA), in which the empirical average over a single chain is used to approximate expectations with respect to the invariant distribution, without a Metropolis-Hastings correction.

While standard ULA has proven to be a valuable tool in simulation, it often exhibits slow convergence. Many methods have been proposed to improve the convergence properties of standard ULA. One such class of methods involves simulating from appropriate reversible and irreversible perturbations of the original Markov process, e.g., \cite{girolami2011riemann,hwang1993accelerating,hwang2005accelerating,lelievre2013optimal,rey2015irreversible,rey2015variance,rey2016improving,zhang2022geometry,cui2024optimal}. Such methods have been shown to be effective in accelerating convergence while guaranteeing improved performance of the corresponding ergodic estimator; see Section \ref{subsec:perturbations}.

A different class of sampling methods that has been explored in the literature employs transport maps; see, for instance, \cite{el2012bayesian,brennan2020greedy,kobyzev2020normalizing,marzouk2016introduction,peherstorfer2019transport,rezende2015variational,papamakarios2021normalizing,tabak2010density}. Transport maps provide functional representations of complex random variables by transforming them, deterministically and invertibly, to a simple reference random variable of the same dimension---for example, a standard Gaussian. We discuss these methods in Section \ref{subsec:transportmaps}. An ``exact'' transport map $S_\text{exact}$ achieves this measure transformation exactly, i.e., it pushes forward a complex target distribution $\pi$ to a desired reference distribution $\eta_\text{ref}$, a relationship we denote as $(S_\text{exact})_\sharp \pi = \eta_\text{ref}$. If such a map is constructed, one can efficiently sample from the target $\pi$ by evaluating $S_\text{exact}^{-1}$ on samples drawn from the reference distribution.

A crucial starting point for this work is the observation that no matter what form of transport map $S$ and construction procedure are chosen, the pushforward $S_\sharp \pi$ obtained in practice will generally \textit{\textbf{not} be equal to the {desired} reference distribution}---due to limitations of the (finite) approximation space in which the map is sought, the finite number of samples used for estimation, behavior of the associated numerical optimization procedure, etc. We demonstrate this idea in Figure \ref{fig:CartoonPushForwardReference}, which represents an example analyzed in detail in Section \ref{SS:FunnelDensity}.
In this case, using the transport map $S$ directly for sampling will lead to bias, as $S^{-1}_\sharp \eta_\text{ref}$ will differ from the target measure $\pi$. As we will see below, combining any such \textit{approximate map} with Langevin dynamics can mitigate this bias. We will also see that approximate maps can \emph{accelerate} the convergence of Langevin sampling for $\pi$.

    \begin{figure}[h]
  \centering
\includegraphics[trim = 100 120 150 120, clip, width=0.7\linewidth]{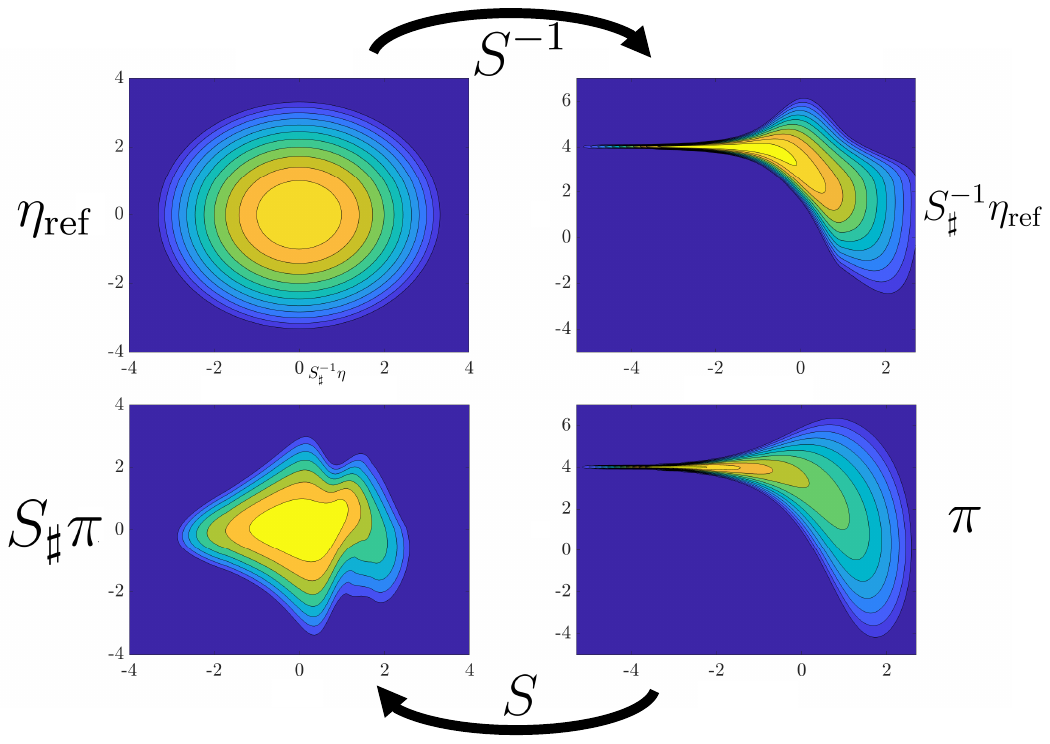}
  \caption{The pushforward distribution $S_\sharp \pi$ obtained with any approximate map $S$ (bottom left) is not equal to the \emph{desired} reference distribution $\eta_\text{ref}$ (top left). This is a cartoon representation of the example studied in Section \ref{SS:FunnelDensity}.}
  \label{fig:CartoonPushForwardReference}
\end{figure}

\bigskip

\noindent
\textit{Our contribution.}
The contribution of this paper is threefold. First, we rigorously demonstrate that applying a transport map to standard overdamped Langevin dynamics (OLD) corresponds to a Riemannian manifold Langevin dynamics (RMLD) \cite{girolami2011riemann}, which is equivalent to a reversible perturbation of the OLD \cite{rey2016improving}. Moreover, we show that applying a transport map to an irreversibly-perturbed OLD system results in a geometry-informed irreversible perturbation (GiIrr) \cite{zhang2022geometry} of the original dynamics. The resulting reversible or irreversible dynamics are naturally linked to the chosen transport map in a precise way; see Section \ref{subsec:connections}. A key consequence of this link is a \emph{systematic} way of learning reversible perturbations---by learning maps.

Second, while transport map transformations of the OLD are equivalent to RMLD or GiIrr in continuous time, we show that when \emph{discretized}, the resulting discrete-time Markov chains have different properties. Specifically, we present the transport map unadjusted Langevin algorithm (TMULA), which is the stochastic process produced by applying a transport map to an Euler-Maruyama discretization of the OLD. The differences between a TMULA and the process produced by a direct discretization of RMLD depend on how strongly the transport map departs from linearity. We discuss these discretization issues in Section \ref{subsec:tmrmld}. A schematic representation of these observations is in Figure \ref{fig:CartoonContinuousDiscrete}; see Section \ref{sec:background} for notational details.
\textcolor{black}{The simulation studies of Section \ref{sec:numerical} demonstrate that TMULA typically produces better samplers than the discretization of the equivalent RMLD or GiIrr.}
\begin{figure}[h!]
  \centering
  \includegraphics[width = \textwidth,trim = 30 90 10 90, clip]{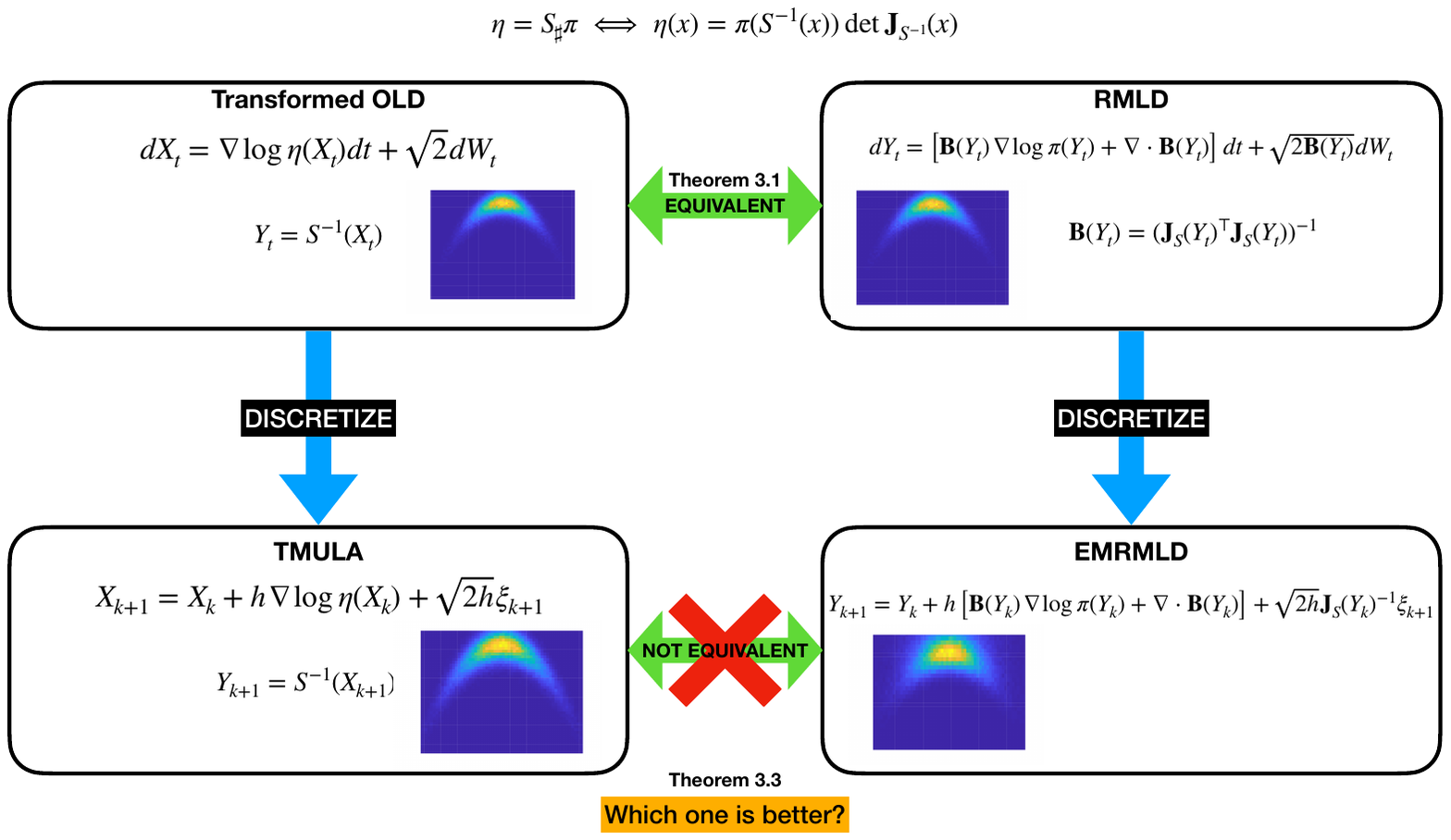}
  \caption{TMULA and RMLD are equivalent in continuous time, but not necessarily after discretization. See Section \ref{subsec:tmrmld} for details. The banana-shaped density in each panel corresponds to each algorithm discussed in Section \ref{SS:BananaExample}. The bottom left panel shows that discretizing TMULA produces a more accurate approximation of the target than discretizing RMLD.}
  \label{fig:CartoonContinuousDiscrete}
\end{figure}

Third, we prove theoretical guarantees of convergence of TMULA in the $2$--Wasserstein metric under the assumptions that (a) the transport map is sufficiently monotone (Assumption \ref{ass:mapmonotone}) and (b) the logarithm of the pushforward of the target measure under the chosen transport map is strongly convex and has Lipschitz gradients (Assumption \ref{ass:potentialass}); see Section \ref{subsec:tmulaanalysis}. In Remark \ref{R:AdvantageUsingTM}, we demonstrate that even when one samples from a log-concave distribution with ULA, it is often practically advantageous to use a transport map when available. Moreover, our theory provides guidance on the ideal choice of reference measure when constructing maps. \textcolor{black}{However, we emphasize that in this paper we do not address the general question of how to judiciously construct transport maps to optimize sampling performance. This question and some others that are left for future work are summarized in Section \ref{S:Conclusion}}.

\textbf{Related work.}
To set these contributions in context, let us discuss the connections between TMULA and other mapped Langevin processes. The mirrored-Langevin \cite{hsieh2018mirrored} and mirror-Langevin dynamics (MLD) \cite{chewi2020exponential,zhang2020wasserstein} are particular versions of TMULA in which the transport map is defined as the gradient of a convex scalar function. MLD \cite{hsieh2018mirrored} is most similar to the dynamics we study here, as the reference process is a Langevin process forced by an isotropic Brownian motion, while the MLD in \cite{zhang2020wasserstein} has a state-dependent diffusion matrix. The transformed unadjusted Langevin algorithm (TULA) is another example of a mapped Langevin process \cite{he2022heavy}, concerned with sampling heavy-tailed distributions. This algorithm uses specific isotropic maps, inspired by the approach in in \cite{johnson2012variable} for random-walk Metropolis samplers, to transform the heavy-tailed target distribution to have sub-exponential tails. Standard ULA is then used to sample from the transformed distribution. In \cite{cui2024optimal}, the \emph{optimal} reversible perturbation is characterized as the minimizer of a Riemannian Poincar\'e constant. The existence of such a minimizer is provided in terms of a so-called moment map.

MLD, TULA, and the design of an optimal Riemannian metric are cases of TMULA when the map or reversible perturbation is defined \emph{a priori}, often for the specific purpose of transforming the support or tails of the target distribution. Meanwhile, a great deal of recent work has focused on \emph{learning} a flexible and expressive transport map from an unnormalized density and/or from samples. For example, normalizing flows---popular in both variational inference and generative modeling---define an invertible map via a composition of simple transformations, parameterized by neural networks; these maps can be learned using evaluations of the unnormalized target density \cite{rezende2015variational}, using samples \cite{tabak2010density,kobyzev2020normalizing,papamakarios2021normalizing}, or both \cite{kohler2020equivariant}. Other classes of maps used in practice include triangular maps \cite{baptista2020adaptive,marzouk2016introduction} or approximations of optimal transport maps \cite{onken2021ot}.

We emphasize that the results in this paper are not tied to any particular class of monotone maps or map construction algorithm. Our goal is to study the impact of a transport map (e.g., one that approximately normalizes a target distribution) on perturbations and discretization of the overdamped Langevin process for that distribution. In fact, one can envision many algorithms that include TMULA as a core element, in which a map is learned either \emph{a priori} via samples or density evaluations, or instead learned and refined on-the-fly \cite{parno2018transport,gabrie2022adaptive}, and then applied to an overdamped Langevin process. Our numerical studies employ a specific computationally tractable class of monotone triangular maps that approximate the Knothe--Rosenblatt rearrangement \eqref{Eq:KR_structure} (see Section \ref{subsec:transportmaps} and \cite{carlier2010knothe,marzouk2016introduction,baptista2020adaptive}  for a quick review), but many other numerical choices are possible.

\medskip
The rest of the paper is organized as follows. In Section \ref{sec:background}, we review reversible and irreversible perturbations to ULA as well as transport maps. Section \ref{sec:tmula} contains our main results on TMULA, i.e., continuous time equivalence to RMLD or GiIrr, their differences when discretized, and the theoretical results on convergence in Wasserstein metric; proofs are deferred to Appendices \ref{app:connections} and \ref{app:wasser}. In Section \ref{sec:numerical}, we present four simulation studies to showcase the performance of TMULA. In Section \ref{S:Conclusion}, we discuss some future avenues for research.

\small
\subsection*{Glossary of acronyms}
For readability, we provide a glossary of the most frequently used acronyms in this paper in Table \ref{table:glossary}.
\begin{table}[H]
\begin{tabular}{>{\raggedright}p{1.5cm}| p{8cm}| p{2cm}}
Acronym & Meaning & Reference \\
\hline
OLD     & Overdamped Langevin dynamics & \eqref{eq:old}\\[5pt]
ULA     & Unadjusted Langevin algorithm & \eqref{eq:dld}, \eqref{eq:ula}\\[5pt]
RMLD    & Riemannian manifold Langevin dynamics & \eqref{eq:revperturb}\\[5pt]
GiIrr   & Geometry-informed irreversible perturbation & \cite{zhang2022geometry} \\[5pt]
TMRMLD  & Transport map defined RMLD & Thm~\ref{prop:TMRMLD} \\[5pt]
TMULA   & Transport map unadjusted Langevin algorithm &  Sec~\ref{subsec:transformingLD}, \eqref{eq:onesteptmula}\\[5pt]
EMRMLD  & Euler-Maruyama discretization of TMRMLD & \eqref{eq:onestepemrmld} \\[5pt]
ATM     & Adaptive transport maps & \cite{baptista2020adaptive}, \cite{baptista2022atm} \\ [5pt]
KSD & Kernelized Stein discrepancy & \cite{GorhamKSD} \\ [5pt]
UILA    & Unadjusted implicit Langevin algorithm & Sec~\ref{SS:RosenbrockDistribution}, \cite{casella2011stability} \\[5pt]
TMUILA  & Transport map unadjusted implicit Langevin algorithm & Sec~\ref{SS:RosenbrockDistribution}, \eqref{eq:splitstep}
\end{tabular}
\centering
\caption{Glossary of acronyms used in this paper in order of appearance.}
\label{table:glossary}
\end{table}
\normalsize

\section{Langevin samplers and transport maps}
\label{sec:background}
We review the overdamped Langevin dynamics (OLD) and techniques for accelerating their convergence via reversible and irreversible perturbations. We also provide an overview of measure transport.%

\subsection{Overdamped Langevin dynamics and the unadjusted Langevin algorithm (ULA)}
\label{subsec:ula}
Let $Y$ be a $\R^d$--valued random variable with unnormalized target density $\pi$.
A stochastic differential equation (SDE) that has $\pi$ as its invariant density is the so-called OLD:
\begin{align}
  \de Y(t) = \nabla \log \pi(Y(t)) \de t + \sqrt{2} \de W(t). \label{eq:old}
\end{align}
Here,  $W(t)$ is a standard Brownian motion on $\R^d$. By ergodicity, expectations of test functions $\phi:\R^d\to \R^d$ with respect to $\pi$ can be approximated by the time average of $\phi(y)$ over a single trajectory, e.g., \cite[Section 2.4]{pavliotis2014stochastic}:
\begin{align}
  \Ex_\pi[\phi(Y)] = \int_{\R^d} \phi(y) \pi(y) \de y = \lim_{T\to\infty} \frac{1}{T} \int_0^T \phi(Y(t)) \de t. \label{eq:ergodic}
\end{align}
To practically use this relation to estimate expectations with respect to $\pi$, the Langevin dynamics must be discretized. An Euler-Maruyama discretization of \eqref{eq:old} yields
\begin{align}
  Y_{k+1} = Y_k + h\nabla\log \pi(Y_{k+1}) + \sqrt{2h} \xi_{k+1}, \label{eq:dld}
\end{align}
where $\xi^{k+1} \sim \mathcal{N}(0,\mathbf{I})$ are independent, $h$ is the step size, and $Y_k = Y(kh)$; see \cite{kloeden1992stochastic} for a general overview of SDE discretization schemes. The expectation is then approximated by
\begin{align}
  \Ex_\pi[\phi(Y)] \approx \frac{1}{K}\sum_{k = 0}^{K-1} \phi(Y_k).\label{eq:ula}
\end{align}
The algorithm computing this estimator is known as the \emph{unadjusted Langevin algorithm} (ULA), and has been used in problems ranging from molecular simulation to Bayesian inference \cite{durmus2019high}. While the relation in \eqref{eq:ergodic} is exact, discretization introduces bias into the estimator. Some previous research focuses on removing this bias by considering the \emph{Metropolis-adjusted Langevin algorithm} (MALA), in which the Langevin dynamics is used as a proposal distribution within the Metropolis-Hastings algorithm \cite{roberts1996exponential,girolami2011riemann}. Recently, there has been renewed interest in ULA, especially for high-dimensional problems in machine learning. In particular, \cite{durmus2017nonasymptotic,durmus2019high} provide a collection of convergence results for ULA on strongly log-concave distributions. They derive inequalities showing fast convergence of ULA to the invariant distribution in the 2-Wasserstein and total variation distances, and use those results to obtain bounds on the mean-squared error of ULA for Lipschitz and bounded test functions \cite{durmus2019high}. Related work extends these results to distributions that satisfy log-Sobolev inequalities \cite{chewi2021analysis,erdogdu2021convergence,vempala2019rapid}.

\subsection{Accelerating Langevin-based samplers via perturbations} \label{subsec:perturbations}
The OLD in \eqref{eq:old} is not the only stochastic dynamical system that has $\pi$ as its invariant density. In fact, there are infinitely many such systems, some of which can converge to $\pi$ more quickly than OLD \cite{pavliotis2014stochastic}. The work of \cite{hwang2005accelerating,rey2015irreversible,rey2015variance,rey2016improving} show that appropriately chosen reversible or irreversible perturbations to the OLD can accelerate its convergence to the invariant distribution.

\subsubsection{Reversible and irreversible perturbations}\label{SS:ReversiblePerturbations}
Let $\mathbf{B}$ be a matrix-valued valued function that maps $y\in\R^d$ to symmetric positive definite matrices $\mathbf{B}(y) \in \R^{d\times d}$. A reversibly perturbed overdamped Langevin dynamics is given by
\begin{align}
  \de Y(t) = \left[  \mathbf{B}(Y(t)) \nabla \log \pi(Y(t)) + \nabla \cdot \mathbf{B}(Y(t)) \right] \de t + \sqrt{2 \mathbf{B}(Y(t))} \de W(t), \label{eq:revperturb}
\end{align}
where  $(\nabla \cdot \mathbf{B}(y))_i = \sum_{j = 1}^d \frac{\partial \mathbf{B}_{ij}(y)}{\partial y_j}.$ In \cite{rey2016improving}, the authors show that if $\mathbf{B}(y) - \mathbf{I} \succ \mathbf{0}$, then accelerated convergence to the invariant density will be achieved. The degree of accelerated convergence depends on the choice of $\mathbf{B}$, but, in general, there is not a known optimal choice of $\mathbf{B}$. One well-studied reversible perturbation is the \emph{Riemannian manifold Langevin dynamics} (RMLD), in which the matrix $\mathbf{B}(y) = \mathbf{G}(y)^{-1}$ is defined in terms of the metric of a Riemannian manifold. In particular, when the target distribution is the posterior of a Bayesian inference problem, \cite{girolami2011riemann} chooses the metric $\mathbf{G}(y)$ to be the expected Fisher information matrix plus the Hessian of the log-prior. We refer the reader to \cite{livingstone2014information,xifara2014langevin} for further discussion of RMLD and related MCMC methods and to \cite{trillos2020bayesian} for an exploration of Bayesian inference from an optimization perspective.

With regard to irreversible perturbations, consider the SDE
\begin{align}
  \de Y(t) = \left[ \nabla \log \pi(Y(t)) + \psi(Y(t)) \right] \de t + \sqrt{2} \de W(t). \label{eq:irrev}
\end{align}
An irreversible perturbation corresponds with choosing the term $\gamma$ such that the resulting SDE is not time-reversible, while preserving the invariant distribution of the unperturbed system. By using the Fokker--Planck equation of \eqref{eq:irrev}, one can show that to preserve the invariant measure, $\psi$ should be chosen such that $\nabla \cdot (\psi(y) \pi(y)) = 0$ for all $y$. A simple and frequently used choice is $\psi(y) = \mathbf{C}\nabla \log \pi(y)$, where $\mathbf{C}$ is a constant skew-symmetric matrix, i.e., $\mathbf{C} = -\mathbf{C}^\top$. This irreversible perturbation is computationally attractive, as applying it requires only one additional matrix-vector multiplication for each step of the Langevin dynamics. The work \cite{zhang2022geometry} proposes an irreversible perturbation tailored to an already reversibly-perturbed OLD, and empirically shows the advantages of using a state-dependent skew-symmetric matrix in the irreversible perturbation. The resulting so-called \emph{geometry-informed irreversible perturbation} is $\psi(y) = \mathbf{C}(y)\nabla \log \pi(y) + \nabla\cdot\mathbf{C}(y)$, where $\mathbf{C}(y) = -\mathbf{C}(y)^\top$. Theoretical results show that, in continuous time, the performance of the Langevin sampler cannot be hurt by the choice of the irreversible perturbation; see \cite{hwang2005accelerating,rey2015irreversible,rey2015variance,rey2016improving}. We also refer the interested reader to the works of \cite{lelievre2013optimal,ottobre2019optimal,zhang2022geometry} for further results related to irreversible perturbations.

\subsection{Transport maps}
\label{subsec:transportmaps}
Transport maps arise from deterministic couplings of probability measures; see \cite{marzouk2016introduction,villani2009optimal} for an overview. Let smooth probability densities $\eta$ and $\pi$ correspond to probability measures $\mu_\eta$ and $\mu_\pi$, respectively. %Assume $\eta$ and $\pi$ are smooth.
A coupling of the two measures is a joint measure $\zeta$ whose marginals are $\mu_\eta$ and $\mu_\pi$. A coupling is deterministic if there exists a measurable function $T$, called a transport map, such that for all {$\mu_{\pi}$-}measurable sets $A$, $\mu_\pi(A) = \mu_\eta(T^{-1}(A))$ \cite{villani2009optimal}. This is denoted $\mu_\pi = T_\sharp \mu_{\eta}$,\footnote{We will also abuse notation by writing $\pi = T_\sharp \eta$.} and $\mu_\pi$ is the pushforward of $\mu_\eta$ through the map $T$. If $T$ is invertible, then $\pi(y) = \eta(T^{-1}(y)) \det \mathbf{J}_{T^{-1}}(y)$, where $\mathbf{J}_T$ is the Jacobian of $T$.

Transport maps provide a method for expressing complex random variables in terms of simpler ones. For example, if we choose $X\sim \eta \equiv \mathcal{N}(0,\mathbf{I})$, and construct a map $T: \R^d\to \R^d$ such that $\pi = T_\sharp \eta$, then we can produce independent realizations of random variable $Y\sim \pi$ by drawing independent realizations of the standard normal and evaluating the transport map at the sample points, i.e., $T(X) \sim \pi$. In general, there are infinitely many maps that can induce a deterministic coupling of $X$ and $Y$ \cite{villani2009optimal}, and there are many computational methods for constructing them, e.g., via optimal transport \cite{el2012bayesian,onken2021ot,manole2021plugin}, triangular transport \cite{baptista2020adaptive,brennan2020greedy,peherstorfer2019transport}, normalizing flows \cite{kobyzev2020normalizing,papamakarios2021normalizing,tabak2013family}, and so on.

\subsubsection{Monotone triangular transport maps}
\label{sec:triangular}
While the sampling algorithms and theoretical results of this paper will apply to \emph{any} sufficiently smooth and invertible transport map, our numerical illustrations will employ a specific, computationally tractable class: {triangular monotone maps}. Let $Y = (y_1,\ldots,y_d)$ and consider a map $S$ of the form
\begin{align}\label{Eq:KR_structure}
  S(Y) = \begin{bmatrix*}[l]
    S_1(y_1) \\
    \, \, \vdots \\
    S_d(y_1,y_2,\ldots,y_d)
  \end{bmatrix*},
\end{align}
where the $i$th component of the map is monotone increasing in $y_i$, i.e., $\frac{\partial S_i}{\partial y_i} >0, \, \forall i$. This monotonicity condition ensures that $S$ is invertible, i.e., that there exists a map $T=S^{-1}$. A triangular monotone map satisfying $S_\sharp \mu_{\pi} = \mu_{\eta}$ is guaranteed to exist when $\mu_{\pi}$ and $\mu_{\eta}$ are atomless (which is always the case in this paper, as we assume both measures to have Lebesgue densities on $\mathbb{R}^d$). A map $S$ of the form \eqref{Eq:KR_structure} is often called a Knothe--Rosenblatt (KR) rearrangement; see \cite{rosenblatt1952remarks,bogachev2005triangular,carlier2010knothe,marzouk2016introduction}.

Enforcing triangular structure offers several advantages. First, the pushforward density of $\pi$ through $S$ can be easily computed since the determinant of the Jacobian of $S$ is the product of diagonal entries. Second, the triangular structure of the map allows the inverse to be easily computed.  Evaluating $S^{-1}(x)$ involves solving a sequence of $d$ one-dimensional root finding problems, and since the $i$th entry of each map component is monotone with respect to $y_i$, the roots are unique. Third, triangular maps enjoy certain sparsity and decomposability properties that follow from the conditional independence structure of $\pi$, particularly when $\eta$ is chosen to be a product measure. These properties are useful when building transport maps in high dimensions, but do depend on the ordering of the variables in the triangular construction; see \cite{spantini2018inference} for details.  We refer the interested reader to \cite{baptista2020adaptive} for a construction of triangular maps that is based only on samples of $\pi$ and guarantees monotonicity everywhere.

\textcolor{black}{
Most methods for computing monotone triangular maps minimize the Kullback--Leibler (KL) divergence between the target $\pi$ and approximating distribution $T_\sharp \etaref$, for some fixed choice of $\etaref$; the direction of the KL divergence depends on what information is available. \cite{el2012bayesian} constructs maps given unnormalized evaluations of the density $\pi$, by minimizing  $D_{\text{KL}}( T_\sharp \etaref \| \pi)$ plus a penalty term chosen so that the minimizer is monotone with high probability; the map is represented in a linear space (for instance in the span of a chosen set of polynomials). \cite{baptista2020adaptive}, in contrast, introduces a triangular map parameterization that guarantees monotonicity everywhere, and shows how to construct the transport map given only samples of $\pi$. The $i$-th component function of the map $S$ \eqref{Eq:KR_structure} is chosen to be of the form $S_i(y_1,\ldots,y_i ) = f(y_1,\ldots,y_{i-1},0) + \int_0^{y_i} g(\partial_{i} f^i(y_1,\ldots,y_{i-1},t)) \de t$,
where $f^i: \R^i \to \R$ is represented in a polynomial or wavelet basis and $g: \R \to \R_{> 0}$ is a bijective and strictly positive function, such as a softplus or shifted exponential linear unit. The parameters of the map are then learned by minimizing an empirical approximation of $D_{\text{KL}}(\pi \| S^{-1}_\sharp \eta)$. This is equivalent to maximum likelihood estimation of the map: given samples $\{Z^i\}_{i = 1}^N \sim \pi$, we find $\hat{S} = \arg\max \sum_{i = 1}^N \log S^{-1}_\sharp \eta(Z^i)$, where $\eta$ is a standard Gaussian density and the optimization is over the space of functions $\{f^i\}_{i=1}^d$ in the monotone parameterization described above. Crucially, under appropriate conditions on $g$, this optimization problem is smooth and has \emph{no spurious local optima}. Moreover, \cite{baptista2020adaptive} describes a greedy algorithm for enriching the function spaces in which one seeks each $f^i$, yielding a semi-parametric estimation procedure. We employ this approach to build the maps used in our numerical studies, with code available in the \texttt{ATM} GitHub repository \cite{baptista2022atm}.}

\section{Transport map unadjusted Langevin algorithm (TMULA)}
\label{sec:tmula}

Now we present our main methodology, TMULA. The main idea behind TMULA is to apply a transport map to the target distribution so that the pushforward of the target through the map becomes easier to sample using ULA. We show that, in continuous time, this construction can be understood as creating a reversible perturbation to the original OLD, and that it also naturally gives rise to geometrically-informed irreversible perturbations. Comparing TMULA to a direct discretization of the equivalent RMLD, we show that TMULA provides a different and in some cases more accurate approximation of the desired target measure. Finally, building on recent analysis of ULA, we show that in discrete time, under appropriate conditions on the map and target, TMULA converges geometrically to the target distribution in $\mathcal{W}_2$ and that the rate of convergence can be accelerated relative to the original OLD.

\subsection{Transforming Langevin dynamics}
\label{subsec:transformingLD}
Suppose we are given a continuously differentiable unnormalized target density $\pi$ with support on $\R^d$ and an invertible twice-continuously differentiable map $S: \R^d \to \R^d$. TMULA consists in applying ULA to the resulting pushforward density $\eta \coloneqq S_\sharp \pi$.\footnote{Following the discussion in Section~\ref{sec:intro}, we note that in general $\eta$ will differ from the ideal desired reference $\etaref$ used in any given map construction procedure. Once a map is constructed, the desired reference is no longer directly relevant; we work only with $\pi$, $S$, and the resulting $\eta$ from the numerical algorithm.}
Samples from $\pi$ are then obtained by applying the inverse map $T = S^{-1}$ to the trajectory produced by such an ULA.

Different choices of the transport map $S$ yield different versions of TMULA. Recall that $\eta(x) = (S_\sharp\pi)(x) = \pi(T(x))\det \mathbf{J}_{T}(x)$, where $\mathbf{J}_T(x)$ is the Jacobian of $T$. Computing the gradient of $\log\eta(x)$ requires computing the gradient of the log determinant of $S$. We have
\begin{align*}
  \nabla_x \log \eta(x) = \nabla_x(\log \pi(T(x)) \det \mathbf{J}_T(x)) = \mathbf{J}_T^\top  \nabla_y \log \pi(T(x)) - \mathbf{J}_{T}^\top\nabla_y \log \det \mathbf{J}_S(y),
\end{align*}
where $y=T(x)$. While in principle the map $S$ only needs to be invertible and twice continuously differentiable, for computational purposes $S$ should be chosen so that $\nabla_y \log \det \mathbf{J}_S(y)$ can be computed efficiently. For example, \cite{hsieh2018mirrored} defines $S$ as the gradient of a convex scalar function, but only considers cases in which the determinant of the Hessian of this function can be computed exactly. If $S$ is constructed with a normalizing flow \cite{kobyzev2020normalizing,papamakarios2021normalizing}, there are many parameterizations that allow efficient computation of the log determinant. Triangular maps, as described in Section~\ref{sec:triangular}, in general allow for fast computation not only of the log determinant but also its gradient. In the triangular case, we have $\log \det \mathbf{J}_T(x) = -\log\det \mathbf{J}_S(y) = - \sum_{i = 1}^d \log \frac{\partial S_i}{\partial y_i}$, and so
\begin{align}
  \nabla_x \log \eta(x) =   \mathbf{J}_S^{-\top}(T(x))\left(  \nabla_y \log \pi(T(x)) - \sum_{i = 1}^d \left( \frac{\partial S_i}{\partial y_i} (T(x)) \right)^{-1} \mathbf{H}_i(T(x)) \right), \label{eq:triangulargradlogdet}
\end{align}
where $\mathbf{H}_i(y) =  \left[ \frac{\partial^2 S_i}{\partial y_1\partial y_i}(y) , \ldots, \frac{\partial^2S_i}{\partial y_d \partial y_i}(y) \right]^\top$; see also \cite{parno2018transport}.

Regardless of the particular choice of the invertible map, applying ULA to $\eta = S_\sharp \pi $ with step size $h$ yields a discrete-time Markov chain,
\begin{align}
  X_{k+1} = X_k + h\nabla_x \log \eta(X_k) + \sqrt{2 h} \xi_{k+1}, \, \xi_{k+1} \sim \mathcal{N}(0,\mathbf{I}),\nonumber
\end{align}
which produces trajectories that approximately sample from $\eta$, while $Y_k = T(X_k)$ will produce trajectories that approximately sample $\pi$.

The idea of transforming a Langevin sampler is not new. Indeed, there have been many instances of using some type of invertible mapping to ``reshape'' the target distribution so that sampling schemes, such as ULA or MCMC, will perform better. For example, \cite{johnson2012variable} proposes a class of isotropic invertible transformations for sub-exponentially light densities that makes them super-exponentially light. The random-walk Metropolis algorithm can then be shown to be geometrically ergodic for the transformed densities. In the same spirit, \cite{he2022heavy} uses the same transformations for heavy-tailed densities but focuses on the convergence properties of ULA. \cite{hsieh2018mirrored,chewi2020exponential,zhang2020wasserstein} consider Langevin sampling for densities with bounded support, transforming to densities supported on $\R^d$. These papers consider specific maps defined as gradients of convex functions, so that the inverse map is given by a Legendre transformation. %

In the efforts just described, however, the map $S$ is not tailored to the target density at hand, other than to its support or the decay rate of its tails. In contrast, we are interested in more general situations where an invertible map (in general an \textit{approximate} map) \emph{can be} learned and tailored to reshape the entire density. Recent literature has seen many algorithmic realizations of this idea. \cite{brennan2020greedy} learns both normalizing flows and triangular maps from unnormalized target densities, and uses these maps to transform the target density before applying MCMC. \cite{peherstorfer2019transport} takes a similar approach, but proposes using a computationally cheaper approximation of the target density to learn the map, before performing MCMC on a transformation of the true target density. Other work seeks to learn and update a map on-the-fly using samples produced during MCMC: \cite{parno2018transport} does so with triangular maps, and proves ergodicity of the resulting adaptive MCMC approach; \cite{gabrie2022adaptive} pursues a similar construction with normalizing flows. These efforts have all demonstrated strong empirical performance, but so far there is little theoretical understanding of these performance gains. The use of general transport maps with ULA also has not, to our knowledge, been studied.

More fundamentally, we emphasize that details of how the map is learned or constructed are not our focus here. Rather, we are interested in the underlying question of how the performance of ULA is affected by the introduction of a mapping: how this relates to other perturbations of the Langevin dynamics, what finite-sample convergence guarantees can be derived, and what is the impact of a map on asymptotic bias.

\subsection{Connections to perturbations of overdamped Langevin dynamics}
\label{subsec:connections}

One interpretation of TMULA is that the map $S$ ``preconditions'' the target density to improve the performance of ULA. Earlier, we discussed how reversible perturbations can be interpreted as preconditioning the gradient of the log-density so that Langevin dynamics converges faster. We now make a precise connection between these two types of preconditioning in the continuous-time setting. Specifically, we show that the stochastic process induced by OLD on $\eta$, then transformed by $T$, is an RMLD \cite{girolami2011riemann}. Moreover, we also show that a transport map applied to an irreversibly perturbed OLD produces a GiIrr \cite{zhang2022geometry}.
The following results assume only that the map is invertible and twice-continuously differentiable. Proofs are given in Appendix \ref{app:connections}.

\subsubsection{Transport maps induce reversible perturbations}

Let $\pi$ be a continuously differentiable probability density with support on $\R^d$ and let $S$ be an invertible, twice-continuously differentiable map, with $T \coloneqq S^{-1}$. Define the density $\eta$ as the pushforward of $\pi$ under $S$, i.e., $\eta \coloneqq S_\sharp \pi$, and let $X(t)$ be the overdamped Langevin dynamics on $\eta$:
\begin{align}
  \de X(t) &= \nabla \log \eta(X(t)) \de t + \sqrt{2} \de W(t). \label{eq:ou}
\end{align}
\begin{theorem}{(\texttt{TM} + \texttt{LD} = \texttt{TMRMLD})}
\label{prop:TMRMLD}
 The diffusion process $Y(t) = T(X(t))$, where $X(t)$ is the process \eqref{eq:ou}, evolves according to
  \begin{align}
    \de Y(t) = \left[\B(Y(t)) \nabla \log \pi(Y(t)) + \nabla \cdot \B(Y(t)) \right] \de t+ \sqrt{2\B(Y(t))} \de W(t),
    \label{eq:rmld1}
  \end{align}
with $\B(y) = (\J_S^\top\J_S(y))^{-1} = \J_T\J_T^\top(y)$ and $\sqrt{\B(y)} = \J_S^{-1}(y) = \J_T(y)$.
\end{theorem}
This result implies that prescribing a map $S$ from $\pi$ to some $\eta = S_\sharp \pi$ is the same as choosing a reversible perturbation of the OLD for $\pi$. This link provides a systematic way of  building reversible perturbations---beyond, for example, standard choices based on local curvature \cite{girolami2011riemann,martin2012stochastic}---by constructing a  map $S$ that transforms $\pi$ to a desired reference distribution. We will comment in Section \ref{subsec:tmulaanalysis} on what constitutes a ``good'' choice of reference in this setting.
In the other direction, the perspective that certain reversible perturbations correspond to transformed Langevin processes allows us to analyze the quality and convergence properties of such RMLD algorithms by leveraging the analysis of ULA; see again Section \ref{subsec:tmulaanalysis}.

\subsubsection{Transport maps induce geometry-informed irreversibility}

Irreversible perturbations to Langevin dynamics are known to accelerate convergence to the equilibrium distribution by taking advantage of the anisotropy of the target distribution. A natural question to ask is: what is the stochastic process that results from applying a transport map to a reference Langevin dynamics that has an irreversible perturbation? The following Theorem states that the output is a geometry-informed irreversible perturbation of the RMLD derived in Theorem \ref{prop:TMRMLD}.

\begin{theorem}{(\texttt{TM} + \texttt{Irr} = \texttt{TMGiIrr})}
  \label{prop:tmgiirr}
Let $X(t)\in \R^d$ evolve according to
  \begin{align*}
    \de X(t) = (\mathbf{I} + \mathbf{D})\nabla \log \eta(X(t)) \de t + \sqrt{2} \de W(t)
  \end{align*}
  where $\mathbf{D}$ is a constant skew-symmetric matrix, i.e., $\mathbf{D}^{\top}=-\mathbf{D}$. Then the stochastic process $Y(t) = T(X(t))$ evolves according to
  \begin{align*}
    \de Y(t) = \left[ \mathbf{P}(Y(t)) \nabla \log\pi(Y(t)) + \nabla \cdot \mathbf{P}(Y(t))  \right] \de t + \sqrt{2 \mathbf{B}(Y(t))} \de W(t),
  \end{align*}
  with $\mathbf{P}(y) = \mathbf{B}(y) + \mathbf{C}(y)$, $\mathbf{B}(y)$ is defined in \eqref{eq:rmld1}, $\mathbf{C}(y) = \mathbf{J}_T(y) \mathbf{D} \mathbf{J}_T^\top(y)$,   $\mathbf{J}_T(y)=\mathbf{J}_S(y)^{-1} $.%= \mathbf{J}_S(y)^{-1}\mathbf{D}\mathbf{J}_S(y)^{-\top}$.
\end{theorem}

\begin{remark}{(Not all metrics correspond to maps.)}
  Theorem \ref{prop:TMRMLD} showed that a transport map applied to an overdamped Langevin dynamics is equivalent to a Riemannian manifold Langevin dynamics where the metric is related to the Jacobian of the transport map. The converse of this statement, however, is not true. That is, not all RMLD metrics correspond with a transport map. Let us demonstrate this via a simple counterexample. For Bayesian inverse problems, a common heuristic is to choose the metric as the sum of the expected Fisher information matrix and the negative Hessian of the log-prior \cite{girolami2011riemann,martin2012stochastic}. In Section \ref{SS:FunnelDensity}, we use this heuristic, as applied in \cite{girolami2011riemann}, and find that the metric $\mathbf{G}(y_1,y_2)$ is of the form
  \begin{align*}
    \mathbf{G}(y_1,y_2) = \mathbf{B}(y_1,y_2)^{-1} = \begin{bmatrix}
      {2N\beta+e^{y_2}} & 0 \\ 0 & {Ne^{-2{y_2}} + 1/3}
    \end{bmatrix},
  \end{align*}
  where $N,\beta$ are some positive parameters. This implies that the Jacobian is
  \begin{align*}
    \mathbf{J}_S(y_1,y_2)  = \begin{bmatrix}
      \sqrt{2N\beta+e^{y_2}} & 0 \\ 0 & \sqrt{Ne^{-2{y_2}} + 1/3}.
    \end{bmatrix}
  \end{align*}
If there exists a map that has this matrix as its Jacobian, we find that
\begin{align*}
    &\frac{\partial S_1}{\partial y_1} = \sqrt{2N\beta+e^{y_2}} \implies S_1(y_1,y_2) = y_1\sqrt{2N\beta+e^{y_2}} + C_1(y_2) \\
    &\frac{\partial S_1}{\partial y_2} = 0 \implies S_1(y_1,y_2) = C_2(y_1),
\end{align*}
where $C_1(y_2)$ and $C_2(y_1)$ are functions of $y_2$ and $y_1$ alone, respectively. These two conditions on $S_1(y_1,y_2)$ cannot be satisfied simultaneously, and therefore, there exists no map that corresponds with this metric.
\end{remark}

\small
\begin{table}[]
\centering
\begin{tabular}{l|l}
 \\ \hline
 \textbf{\texttt{LD}}& $\de Y_t = \nabla \log \pi(Y_t) \de t+ \sqrt{2} \de W_t $  \\
 \textbf{\texttt{RMLD}}&  $\de Y_t = \left[\B(Y_t)\nabla \log \pi(Y_t)  +\nabla \cdot \B(Y_t) \right] \de t + \sqrt{2\B(Y_t)} \de W_t$\\
 \textbf{\texttt{Irr}} & $\de Y_t = (\mathbf{I}+\mathbf{D}) \nabla \log \pi(Y_t) \de t+ \sqrt{2} \de W_t $  \\
 \textbf{\texttt{GiIrr}} & $\de Y_t = \left[(\B(Y_t)+\bC(Y_t))\nabla \log \pi(Y_t)  +\nabla \cdot (\B(Y_t)+\bC(Y_t)) \right] \de t + \sqrt{2\B(Y_t)} \de W_t$
\end{tabular}
\caption{Summary of the SDEs that share the same invariant density $\pi(y)$. Here, $\mathbf{B}(y) = \mathbf{B}(y)^\top$, $\mathbf{D} = -\mathbf{D}^\top$, and $\mathbf{C}(y) = -\mathbf{C}(y)^\top$}
\label{table:dynamicssummary}
\end{table}
\normalsize

\begin{table}[]
  \centering
  \begin{tabular}{l|l}
  \texttt{TM} +  \texttt{LD} &  \texttt{TMRMLD} \\ \hline
 ${\footnotesize \begin{dcases} \de X_t = \nabla\log  \eta(X_t) \de t + \sqrt{2} \de W_t  \\ Y_t = S^{-1}(X_t)\end{dcases}}$ &  $\footnotesize{\begin{dcases}\de Y_t = \left[\B(Y_t)\nabla \log \pi(Y_t)  +\nabla \cdot \B(Y_t) \right] \de t + \sqrt{2\B(Y_t)} \de W_t \\ \mathbf{B}(Y_t) = (\mathbf{J}_S(Y_t)^\top \mathbf{J}_S(Y_t))^{-1} \end{dcases}}$
  \end{tabular}
  \caption{Summary of Theorem \ref{prop:TMRMLD} stating that these two systems are equivalent. Here, $\eta = S_\sharp \pi$.  } \label{table:tmrmld}
  \end{table}

  \begin{table}[]
      \centering
    \begin{tabular}{l|l}
       \texttt{TM} +  \texttt{Irr} &  \texttt{TMGiIrr} \\ \hline
   ${\scriptsize
   \begin{dcases}
       \de X_t =(\mathbf{I}+\mathbf{D}) \nabla \log \eta(X_t) \de t + \sqrt{2} \de W_t  \\ Y_t = S^{-1}(X_t)\end{dcases}}$
        &  ${\scriptsize
       \begin{dcases}
       \de Y_t = \left[\mathbf{P}(Y_t)\nabla \log \pi(Y_t)  +\nabla \cdot \mathbf{P}(Y_t) \right] \de t + \sqrt{2\B(Y_t)} \de W_t \\ \mathbf{B}(Y_t) = (\mathbf{J}_S(Y_t)^\top \mathbf{J}_S(Y_t))^{-1} \\ \mathbf{P}(Y_t) = \mathbf{B}(Y_t) + \mathbf{J}_S(Y_t)^{-1}\mathbf{D}\mathbf{J}_S(Y_t)^{-\top}\end{dcases}
       }$
       \end{tabular}
        \caption{Summary of Theorem \ref{prop:tmgiirr} stating that these two systems are equivalent. Here, $\eta = S_\sharp \pi$.  } \label{table:tmgiirr}
    \end{table}

\subsection{Comparing discretizations}
\label{subsec:tmrmld}
Now we study discretizations of the perturbed Langevin systems presented in Section \ref{subsec:connections}. Theorem \ref{prop:TMRMLD} showed that there are two equivalent characterizations of the Riemannian manifold Langevin dynamics for $\pi$---either as a transformation $T = S^{-1}$ of an overdamped Langevin system for $S_\sharp \pi$, or as a reversible perturbation on $\pi$ with matrix $(\mathbf{J}_S^\top\mathbf{J}_S)^{-1}$. To simulate \eqref{eq:rmld1}, we can take advantage of these two characterizations and derive two different discretizations, i.e., two different discrete-time Markov chains.

\subsubsection{One-step analysis}
The first discretization we consider is simply TMULA, summarized as follows:
\begin{align}
  Y_{k+1} \coloneqq &  F_{\text{TMULA}}(Y_k,\xi_{k+1})= T\left(X_k + h \textcolor{black}{\nabla_x \log \eta(X_k)} + \sqrt{2h} \xi_{k+1} \right) \label{eq:onesteptmula} \\ \
  = & T\left(S(Y_k) + h \mathbf{J}_S(Y_k) ^{-\top} \left[\nabla_y \log \pi(Y_k)  -\nabla_y \log\det \mathbf{J}_S(Y_k)
  \right] + \sqrt{2h} \xi_{k+1} \right). \nonumber
\end{align}
As an alternative, we consider the discretization given by a direct application of the Euler-Maruyama scheme to \eqref{eq:rmld1}:
\begin{align}
  Y_{k+1} \coloneqq  & F_{\text{EMRMLD}}(Y_k,\xi_{k+1})\label{eq:onestepemrmld} \\
  = & Y_k + h (\mathbf{J}_S^\top\mathbf{J}_S(Y_k))^{-1}\nabla_y \log \pi(Y_k) + h\nabla\cdot(\mathbf{J}_S^\top\mathbf{J}_S(Y_k))^{-1} + \sqrt{2h}\mathbf{J}_S(Y_k)^{-1} \xi_{k+1}. \nonumber
\end{align}
Here, $\xi_{k+1}\sim \mathcal{N}(0,\mathbf{I})$. The next result relates $F_{\text{TMULA}}$ with $F_{\text{EMRMLD}}$.

\begin{theorem}
  \label{prop:tmuladisc}
It holds that
\[F_{\text{TMULA}}(Y_k,\xi_{k+1}) - F_{\text{EMRMLD}}(Y_k,\xi_{k+1}) = h N_k(Y_k,\xi_{k+1}) + \mathcal{O}(h^{3/2}),\]
where the $i$-th component of $N_k(Y_k,\xi_{k+1})$ is a mean-zero non-Gaussian random variable
\begin{align}
N_k^{(i)}(Y_k,\xi_{k+1}) = \xi_{k+1}^\top \nabla^2 T_i(Y_k) \xi_{k+1} - \sum_{j=1}^d \frac{\partial^2 T_i}{\partial x_j^2}(Y_k).
\end{align}
Moreover,
\begin{align*}
  \Ex&\left[ \left\|F_{\text{TMULA}}(Y_k,\xi_{k+1}) - F_{\text{EMRMLD}}(Y_k,\xi_{k+1})\right\|^2\right] \\ &=h^2 \left[\sum_{i,j = 1}^d \left( \frac{\partial^2T_i}{\partial x_j^2}\right)^2 + \sum_{i,j,l = 1}^d \left(\frac{\partial^2 T_i}{\partial x_j\partial x_l} \right)^2 \right] + \mathcal{O}(h^{5/2}).
\end{align*}
\end{theorem}

We prove this result in Appendix \ref{app:tmulaem}. This result shows that these two discretization schemes share the same first moment, while the discrepancies for higher moments depend on the regularity of the transport map $T$. The two discretizations are identical when the transport map has second and higher derivatives equal to zero, corresponding to a linear map. Otherwise, in contrast to EMRMLD, TMULA typically takes \emph{non-Gaussian} steps on the support of $\pi$.

\subsubsection{Approximations of the invariant measure}\label{SS:BiasApproximation}
The ultimate goal of our study of Langevin algorithms is to approximate expectations with respect to $\pi$. We comment on how accurately the numerical invariant distributions of EMRMLD and TMULA approximate the true target distribution. Our discussion here is based on the results of \cite{abdulle2014high}.

Given a discrete-time stochastic process $\{Y_k\}$ with step size $h$ which approximates a diffusion process $\{Y_t\}$ that has invariant measure $\pi(y)$ on $\R^d$, the asymptotic bias of the ergodic estimator for test function $\phi: \R^d \to \R$ is
\begin{align*}
  e(\phi,h) = \lim_{K \to\infty} \frac{1}{K} \sum_{k = 0}^{K-1} \phi(Y_k) - \int_{\R^d} \phi(y) \pi(y) \de y.
\end{align*}
In \cite{abdulle2014high}, it is shown that when the discrete-time process is constructed using a discretization of local weak order $p$, then the asymptotic bias is of the form
 $ e(\phi,h) = -\lambda_p h^p + \mathcal{O}(h^{p+1}),$
where
\begin{align}
  \lambda_p = \int_0^\infty \int_{\R^d} \left(\frac{1}{(p+1)!}\mathcal{L}^{p+1} - A_p \right) u(y,t) \pi(y) \, \de y \,  \de t,\label{Eq:lambdaBias}
\end{align}
$\mathcal{L}$ is the generator of the Langevin process, $A_p$ is a linear operator that depends on the choice of discretization scheme, and $u(y,t)$ solves the Kolmogorov backward equation (KBE) $\partial_t u = \mathcal{L}u$ with initial condition $u(y,0) = \phi(y)$. Both TMULA and EMRMLD are weak order $p = 1$ discretization schemes since the latter is a standard Euler-Maruyama discretization, while the former is an EM discretization of the reference Langevin process that is then mapped forward through transport map $T$. Using this result, one could show that TMULA is a better discretization scheme for constructing ergodic estimators than EMRMLD by showing that $\lambda_1$ for TMULA is smaller in magnitude than that of EMRMLD. While this is generally difficult to do, we will examine a setting where these constants $\lambda_1$ can be computed exactly. Later we will apply this analysis to a particular target $\pi$ and observable $\phi$, and show explicitly that TMULA yields a smaller value.

The following computations requires that map $S$ pushes forward $\pi$ exactly to a standard normal distribution $\eta$. Recall that we are considering a reversibly perturbed Langevin dynamics which has generator  \begin{align*}\mathcal{L} \phi(y) = \langle b(y), \nabla \phi(y)\rangle + \Tr [(\mathbf{J}_S^\top\mathbf{J}_S)^{-1} \nabla^2 \phi(y)],\end{align*} where $b(y) = (\mathbf{J}_S^\top\mathbf{J}_S)^{-1}\nabla \log \pi(y) + \nabla \cdot  [(\mathbf{J}_S^\top \mathbf{J}_S)^{-1}](y)$. Studying the asymptotic bias of TMULA can be reduced to analyzing the ergodic estimator for the test function $\phi\circ T$ for the overdamped Langevin process on $\eta(x)$. Here, $\mathcal{L}^{\text{TMULA}}\phi = \langle \tilde{b},\nabla \phi \rangle + \Delta \phi$, where $\tilde{b}(x) = \nabla \log \eta(x) $; using results from \cite{zygalakis2011existence} we have,
\begin{align*}
  A_1^{\text{TMULA}}\phi = \frac{1}{2}\tilde{b}^\top \left[ \nabla^2 \phi\right] \tilde{b} + \sum_{i = 1}^d \phi^{(3)}(e_i,e_i,\tilde{b}) + \frac{1}{2} \sum_{i ,j = 1}^d \phi^{(4)}(e_i,e_i,e_j,e_j),
\end{align*}
where $\phi^{(3)}$ and $\phi^{(4)}$ are multilinear forms of the third and fourth derivative terms and $e_i$ is the canonical basis in $\R^d$.

Likewise, for EMRMLD, one can derive that
\begin{align*}
  A_1^{\text{EMRMLD}}\phi = &\frac{1}{2}{b}^\top\left[ \nabla^2 \phi\right] {b} + \frac{1}{6} \sum_{i,j,k = 1}^d \frac{\partial^3\phi}{\partial x_i\partial x_j\partial x_k}(b_i\mathbf{Q}_{jk} + b_j \mathbf{Q}_{ik}+b_k\mathbf{Q}_{ij})\\  &+ \frac{1}{8}\sum_{i,j,k,l = 1}^d \frac{\partial^4 \phi}{\partial x_i\partial x_j \partial x_k \partial x_l} \sum_{a,b = 1}^d g_{ia}g_{jb}g_{ka}g_{lb}
\end{align*}
where $\mathbf{Q} = (\mathbf{J}_S^\top \mathbf{J}_S)^{-1}$ and $g_{ij} = (\mathbf{J}_S^{-1})_{ij}$. Since $S$ pushes $\pi$ to a standard normal distribution exactly, the solution to the KBE can be computed exactly via Hermite polynomial expansions.

In Section \ref{SS:BananaExample}, we will show an example where the constants $\lambda_1$ from (\ref{Eq:lambdaBias}) controlling the asymptotic bias of both schemes are analytically computable in this way and verifiable via simulation---demonstrating, for this example, that TMULA has smaller bias than EMRMLD.

\subsection{Convergence in Wasserstein distance}
\label{subsec:tmulaanalysis}

In this section, we visit the convergence guarantees provided for ULA for log-concave distributions. We show that TMULA converges at a geometric rate in the $2$-Wasserstein distance to the target distribution $\pi$ under proper assumptions. The main idea is that the transport map allows us to consider log-concavity of the \emph{pushforward} density $\eta \coloneqq S_\sharp \pi$ and to transfer convergence rates back to the target.
Define $U(x) \coloneqq -\log \eta(x) = - \log S_\sharp \pi(x)$. We make the following assumptions on $S$ and $U$.

\begin{assumption}
  The map $S$ is sufficiently monotone in the sense that there is a $\rho>0$ such that $ \|S(z) - S(z')\| \ge \rho \|z - z'\|$.
  \label{ass:mapmonotone}
\end{assumption}
If $S$ satisfies Assumption \ref{ass:mapmonotone}, then  $T=S^{-1}$ is globally Lipschitz with constant $1/\rho$.
\begin{assumption}
  The function $U(x) = -\log \eta(x)$ is $m$-strongly convex, with $L$-Lipschitz gradients. That is, for all $x,y\in \R^d$, there exist $m$ and $L$ such that
  \begin{align*}
    & U(y) \ge U(x) + \langle \nabla U(x), y-x \rangle + \frac{m}{2} \|x - y\|^2,
     &\|\nabla U(x) - \nabla U(y) \| \le L \|x - y\|.
  \end{align*}
  \label{ass:potentialass}
\end{assumption}
Recall the $2$-Wasserstein distance \cite{hsieh2018mirrored}
\begin{align}
  \mathcal{W}_2^2(\mu,\nu) = \inf_{W:W_\sharp \mu = \nu} \int_{\R^d} \|x - W(x) \|^2 \de \mu.
\end{align}
Denote $\eta^k$ and $\pi^k$ to be the distributions of the discrete-time process $X^k$ and $Y^k$, respectively, at time step $k$. Let the time step $h \in (0,\frac{1}{m+L})$ and $\kappa = \frac{2mL}{m+L}$. Then we have the following result, which follows naturally from Theorem 5 of \cite{durmus2019high}.

\begin{theorem}
Let $S$ and $\eta = S_\sharp \pi$ satisfy Assumptions \ref{ass:mapmonotone} and \ref{ass:potentialass}. Then
\begin{align*}
    \mathcal{W}_2^2(\pi^k,\pi) \le \frac{1}{\rho^2}\left(1- \frac{\kappa h}{2}\right)^k \left(2\|y-y^\star\|^2 + \frac{2d}{m} -C \right) + \frac{C}{\rho^2},
  \end{align*}
  where
  $  C = \frac{2L^2 d}{\kappa }[h(\kappa^{-1} + h)] \left(2+ \frac{L^2h}{m} + \frac{L^2h^2}{6} \right).$
\label{prop:wasser}
\end{theorem}

The proof is in Appendix \ref{app:wasser}. An important consequence of Theorem~\ref{prop:wasser} is that, via the link between transport and reversible perturbations established earlier, we now have \textit{non-asymptotic convergence results} for certain reversibly perturbed \textit{discretized} Langevin dynamics. These non-asymptotic bounds apply to TMULA, rather than to a direct (e.g., Euler--Maruyama) discretization of RMLD; recall by Theorem \ref{prop:TMRMLD} that these two formulations are equivalent in continuous time, but as we have argued their behavior after discretization is not necessarily the same.

The map $S$ in TMULA is a degree of freedom of the method, and can be refined; hence a natural trade-off emerges. A simple transport map might be easy to compute with, but could yield limited acceleration compared to ULA for $\pi$. A complex transport map might add additional computational overhead, but could yield faster convergence by controlling the constants $m$ and $L$ of the pushforward density $\eta$. Along these lines, another useful consequence of Theorem~\ref{prop:wasser} is that it provides guidance for the choice of $S$, which we discuss in the next remark.

\begin{remark}{(What is the optimal $\eta$?)}\label{R:OptimalMapInTheory}
We may gain some insight into what $\eta$ should look like so that maximal acceleration can be achieved according to Theorem \ref{prop:wasser}. The rate of convergence can be optimized by considering the term $r = 1 - \frac{\kappa h}{2}$ alone. Choosing $h$ to be as large as possible, the rate is $ r = 1- \frac{mL}{(m+L)^2}.$
It turns out this rate is optimized (minimizing $r$) only if $\eta$ is an \emph{isotropic Gaussian}. Fix $m>0$, we study the function $r(L) = 1- \frac{mL}{(m+L)^2}$ for $L \in [m,\infty)$. Taking the derivative with respect to $L$ yields $\frac{\partial r}{\partial L} = \frac{m(L-m)}{(m+L)^3},$
which is always positive for $L>m$. Thus, the greatest rate of convergence is attained when $L = m$, which implies that $\nabla^2 U(x) = \mathbf{I}$, and therefore implies that $\eta$ is identically a standard normal. This result supports the intuition that constructing a map $S$ such that $\eta$ is as close as possible to an isotropic Gaussian is best.
\end{remark}

Next we briefly make a useful observation on the relation of our results to target distributions that satisfy a log-Sobolev inequality (LSI). Convergence results (e.g., in Kullback–Leibler (KL) divergence) for ULA have also been developed for target distributions that satisfy LSI, e.g., see \cite{BakryBook, chewi2021analysis,erdogdu2021convergence,vempala2019rapid}. Our goal here is to explore the connection between such target distributions and the transport map $S$. We present a lemma showing that if $S$ satisfies Assumption \ref{ass:mapmonotone} and if its inverse has a globally bounded Jacobian, then if $\eta$ satisfies a LSI, $\pi$ must also satisfy a LSI.  As is known in the literature, see for example \cite{BakryBook}, LSI  implies exponential convergence in KL divergence of the law of the Langevin diffusion that has $\eta$ as its invariant distribution to the target distribution $\eta$.

We recall that a measure satisfies the log-Sobolev inequality with parameter $\alpha>0$ when, for every  smooth function $f:\mathbb{R}^{d} \to \mathbb{R}^{+}$, it holds that
\begin{align*}
\textrm{Ent}_{\eta}[f]&\leq\frac{1}{2 \alpha }\int_{\mathbb{R}^{d}}\frac{\|\nabla f(y)\|^{2}}{f(y)}\eta(y)dy,
\end{align*}
where $\textrm{Ent}_{\eta}[f]=\int_{\mathbb{R}^{d}}f(y)\log f(y)\eta(y)dy-\int_{\mathbb{R}^{d}}f(y)\eta(y)dy\log\left(\int_{\mathbb{R}^{d}}f(y)\eta(y)dy\right).$
In this case, we write that $\eta$ satisfies  $\text{LSI}(\alpha)$.

\begin{lemma}\label{L:LogSobolev}
Assume that the map $S$ satisfies Assumption \ref{ass:mapmonotone} with a differentiable inverse $T=S^{-1}$, $T^{-1}(\mathbb{R}^d)=\mathbb{R}^{d}$, and that the Jacobian of $T$ is uniformly bounded in Frobenius norm,
\begin{align}
\sup_{x\in\mathbb{R}^{d}}\sum_{i=1}^{d}\sum_{k=1}^{d}\left|\frac{\partial T_{k}}{\partial x_{i}}(x) \right|^{2}=\sup_{x\in\mathbb{R}^{d}}\|J_{T}(x)\|_{2}^{2}\leq \frac{1}{\rho^{2}}.\nonumber
\end{align}
Then, if $\eta$ satisfies  $\text{LSI}(\alpha)$, we have that the target measure $\pi = T_\sharp \eta$ satisfies  $\text{LSI}(\alpha \rho^{2})$. \end{lemma}

Lemma \ref{L:LogSobolev} shows that there is no free lunch regarding to exponential convergence to equilibrium. By this we mean that even if we produce an $\eta$ that satisfies a log-Sobolev inequality, the original target measure must also satisfy a log-Sobolev inequality if the map $T=S^{-1}$ has all of its partial derivatives uniformly bounded. If this uniform bound on the Jacobian of $T$ does not hold, however, then the original target measure $\pi$ might indeed not satisfy the log-Sobolev inequality. As we shall see in Remark \ref{R:AdvantageUsingTM} and in Section \ref{sec:numerical}, however, there are advantages of employing a transport map even if $\pi$ satisfies an LSI.

\begin{remark}\label{R:AdvantageUsingTM}
Here we argue that even if the target $\pi$ has nice convexity properties, transforming it with a map $S$ via Theorem \ref{prop:wasser} can improve the convergence rate. Consider a two-dimensional example where the target density is Gaussian and hence strongly log-concave. In particular, let $\pi = \mathcal{N}(0, \text{diag}(1/m, 1/L) )$ for  $m < L$. Then $U(x)  = -\log \pi(x)$ is $m$-strongly convex and $\nabla U$ is $L$-Lipschitz. The map $S$ that transforms $\pi$ to a standard Gaussian is $S(x_1, x_2) = ( \sqrt{m} x_1, \sqrt{L} x_2)$. This map $S$ satisfies Assumption \ref{ass:mapmonotone} with $\rho = \sqrt{m}$.

For the original distribution, since $L > m$, the rate $r = 1 - (\kappa h) / 2$ is not optimal. For the transformed distribution $\eta$ we have $L = m = 1$, and thus we can have $\kappa = 1$ and for $h$ at its maximum value $(h = 1/2)$ we get $r = 3/4$, which is the minimal attainable value according to Remark \ref{R:OptimalMapInTheory}. By Theorem \ref{prop:wasser}, we pay a constant penalty of $1/\rho^2 = 1/m$ multiplying $r^k$. But it is better to make $r$ small in exchange for a constant factor penalty, since doing so maximizes the rate of geometric convergence.
\end{remark}

\section{Numerical examples}
\label{sec:numerical}

We empirically study TMULA samplers and compare them to other Langevin-type samplers, such as ULA and standard discretizations of RMLD.

We first describe how sample quality is measured in the following numerical examples. We compute the bias, variance, asymptotic variance, and mean-squared errors of ergodic estimators for some chosen test functions. Given a test function $\phi(Y): \R^d \to \R$, let $\bar{\phi}_K = \frac{1}{K} \sum_{k = 0}^{K-1} \phi(Y_k)$ be the estimator of $\Ex_\pi[\phi(Y)]$. Here, $\{Y_k\}$ are produced by a single chain of a discretized Langevin process. The asymptotic variance is computed through the method of batch means \cite{asmussen2007stochastic,zhang2022geometry}. The MSE is estimated by computing 100 independent trajectories and then evaluating the relevant observables for each of those trajectories. The true value is either computed analytically, if available, or through a Monte Carlo estimate with $10^8$ samples.
In addition, we also measure sample quality by computing the kernelized Stein discrepancy (KSD).  The KSD is a computable expression that can approximate certain integral
probability metrics for a certain class of functions defined through the action of the Stein
operator on a reproducing kernel Hilbert space; see \cite{GorhamKSD,zhang2022geometry} for details on the implementation of KSD. For the KSD computations, we simulate 100 independent trajectories of a particular system and then compute the KSD using 10000 sample points along those trajectories.

Each example presented below has a different purpose. In Section \ref{SS:BananaExample} we present a simple non-Gaussian example that allows us to explicitly compare the asymptotic bias of an Euler-Maruyama discretization of RMLD with that of TMULA, in a setting where they are equivalent in continuous time, using the theory of Section~\ref{SS:BiasApproximation} and via simulation. In Section \ref{SS:FunnelDensity} we present a funnel distribution that arises from a Bayesian inference problem. The goal of this example is to compare multiple algorithms: ULA, RMLD with standard metrics in previous literature, RMLD with a transport-derived metric, TMULA, and TMULA with irreversible perturbations. In Section \ref{SS:RosenbrockDistribution} we present a more complicated, higher-dimensional example, a ``hybrid Rosenbrock'' distribution. We show in this example that even though ULA is not able to produce samples from the narrow region of the distribution, TMULA can do so when implemented with a split-step implicit Langevin algorithm. Finally, in Section \ref{SS:MultimodalDistribution}, we consider a multimodal distribution. Our goal here is to demonstrate some of numerical challenges that arise when using transport to reshape and sample from multimodal targets.

\subsection{Evaluating discretizations on a simple Banana distribution}\label{SS:BananaExample}
Consider a banana distribution with, up to an additive constant, $\log \pi(y) = -y_1^2/s^2 - (y_2 + by_1^2-100b)^2$. For this example, we can write explicitly a transport map $S$ that pushes the distribution to a standard Gaussian:
\begin{align*}
  S(y_1,y_2) = \begin{bmatrix}
    y_1/s \\ y_2 + b y_1^2 - 100b
  \end{bmatrix}.
\end{align*}
We choose $s = 4$ and $b = 0.01$, and consider the observable $\phi(y_1,y_2) = y_1^2 + y_1 + y_2^2 + y_2$. Using the symbolic algebra toolbox in MATLAB, we find that $\lambda_1^{\text{TMULA}} = -0.62$ while $\lambda_1^{\text{EMRMLD}} = 34.69$, which clearly shows that TMULA results in a better ergodic estimator for this test function. We confirm these results numerically by simulating several chains with total simulation time $T = 10^6$, for a range of step sizes $h$. We show the empirical estimates of the asymptotic bias in Figure \ref{fig:constantslambda}. The slopes of the error plots, which correspond to $\lambda_1$ values in \eqref{Eq:lambdaBias}, are estimated by computing the term $e(\phi,h)/h$, yielding $\lambda_1^{\text{EMRMLD}}\approx 35.26$ and $\lambda_1^{\text{TMULA}} \approx -0.6345$. These values very closely match our theoretical predictions.

\begin{figure}
  \centering
  \includegraphics[width = 0.66\textwidth]{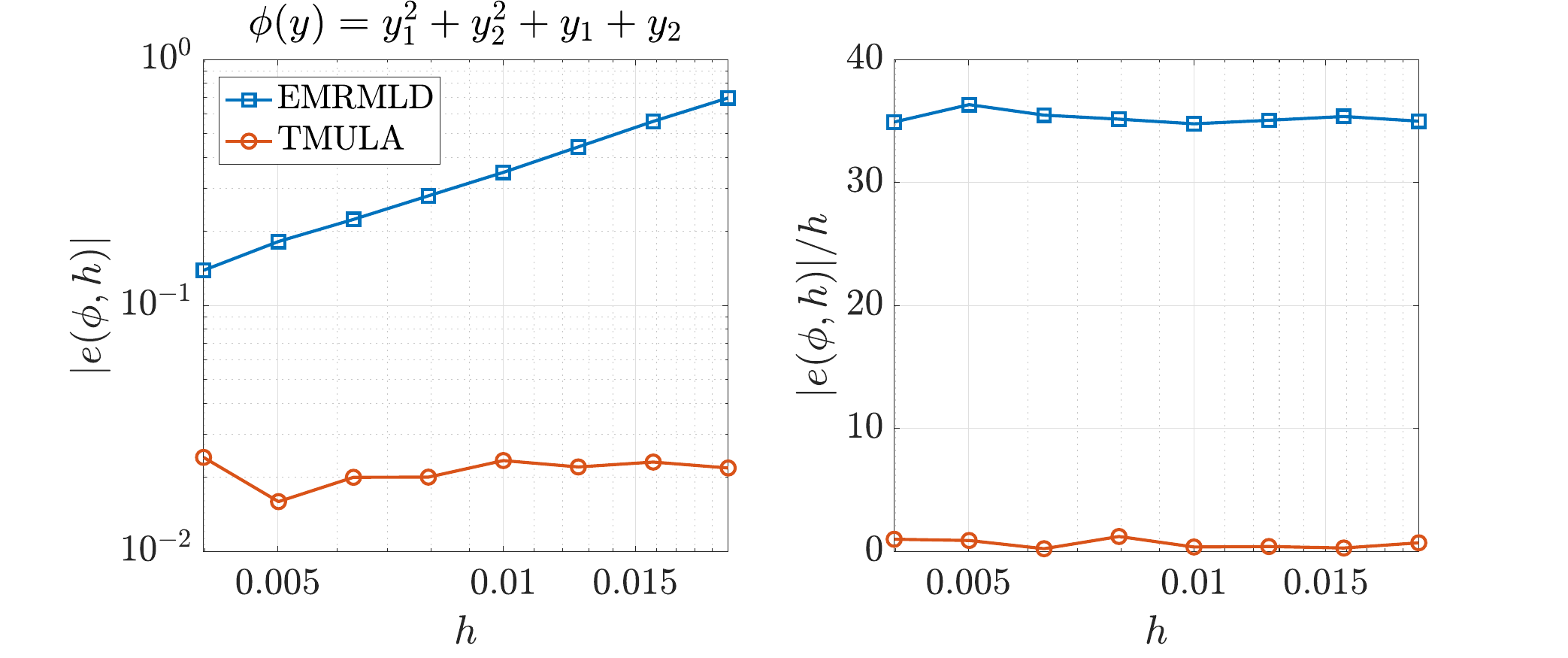}
  \caption{Empirical estimates of $|e(\phi,h)|$ and $e(\phi,h)/h$ for TMULA and EMRMLD as a function of $h$. The constants $\lambda_1$ from \eqref{Eq:lambdaBias} are estimated by averaging over $e(\phi,h)/h$. The empirical estimates match well with the theoretical values of $\lambda_1$ reported in Section \ref{SS:BananaExample}.   }
  \label{fig:constantslambda}
\end{figure}

\subsection{A funnel density: parameters of a normal distribution}\label{SS:FunnelDensity}
We study a Bayesian inference problem similar to the one studied in (\cite{girolami2011riemann}, Section 5). Suppose $X\sim N(\mu,\sigma^{2})$. Given a dataset $\mathbf{X} = \{X_i\}_{i = 1}^N$, we infer the mean and variance.  We assume that the prior on $\mu$ is a normal distribution with $\mu = 0$ and $\sigma^2 = 3$, and that the prior on $\sigma$ is a Gamma distribution with shape parameter $\alpha = 0.75$ and rate parameter $\beta = 0.5$. Furthermore, to make the resulting posterior distribution have infinite support, we consider the variable transformation $\gamma = \log \sigma$. Up to an additive constant, the log-posterior distribution is
\begin{align*}
  \log\pi(\mu,\gamma|\mathbf{X}) = -N\gamma - \frac{1}{2}e^{-2\gamma} \sum_{i = 1}^N (X_i - \mu)^2 - \mu^2/6 + \alpha\gamma - \beta e^{\gamma}.
\end{align*}
One interesting feature about this posterior distribution is that there is a narrow funnel-shaped region that is difficult for ULA to sample from unless a very small step size is used.

We consider five different Langevin dynamics that all have the posterior distribution as their invariant distribution, and compare their sampling performance. They are denoted \texttt{ULA}, \texttt{RMLD}, \texttt{TMULA}, \texttt{TMRMLD}, and \texttt{TMRMLD + Irr}. \texttt{ULA} is the discretization of the standard Langevin dynamics. \texttt{RMLD} is a Riemannian manifold Langevin dynamics where the metric is chosen according to the one in \cite{girolami2011riemann}. In \cite{girolami2011riemann} the authors choose the metric to be the sum of the expected Fisher information matrix and the negative Hessian of the log-prior.
\begin{align*}
  \mathbf{B}(\mu,\gamma) = \begin{bmatrix}
    \frac{1}{2N\beta+e^\gamma} & 0 \\ 0 & \frac{1}{Ne^{-2\gamma} + 1/3}
  \end{bmatrix}.
\end{align*}

Using samples generated by one trajectory of ULA, a monotone triangular $S$ that pushes the target to an approximate normal distribution is learned using \texttt{ATM} with 20000 sample points and where $f$ is trained over a total order Hermite polynomial basis of order 3 that is rectified so that the Jacobian is positive definite (see Section \ref{subsec:transportmaps}) \cite{baptista2020adaptive}. The sample points are generated by simulating a single ULA chain with a step size of $10^{-6}$. Note that this step size is much smaller than the simulations we run to evaluate each system. This is so that ULA produces training samples in the narrow funnel of the distribution.  In \texttt{TMULA} we simulate trajectories of ULA on the pushforward density $S_\sharp \pi$ and map each trajectory into trajectories on $\pi$ through $S^{-1}$. \texttt{TMRMLD} is a RMLD where the metric is chosen according to $\mathbf{B}(\mu,\gamma)^{-1} = \mathbf{J}_S^\top\mathbf{J}_S(\mu,\gamma)$. While the analytical map of this posterior distribution to a standard normal is not known, we show that a metric that arises from the construction of an \emph{approximate} triangular transport map can still result in better sample quality. In continuous time, Theorem \ref{prop:TMRMLD} shows that \texttt{TMULA} and \texttt{TMRMLD} are equivalent, but our analysis in Section \ref{subsec:tmrmld} suggests that these two approaches may yield different results when discretized. We empirically explore their differences. In \texttt{TMULA + Irr} we also consider the discretization of the system in Theorem \ref{prop:tmgiirr} with the same discretization as \texttt{TMULA}.

Figure \ref{fig:funneldensities} shows how the map $S$ can be used to sample from the posterior distribution $\pi$. Observe that when applied directly to a standard normal distribution, the map produces an approximate target distribution that is visually quite different from the target distribution. In contrast, by performing Langevin on the pushforward of the map $S$ on the target distribution, we can produce samples from the posterior. In Figure \ref{fig:ulavstmula} we show sample points produced by \texttt{ULA} and \texttt{TMULA} with the same step size $h = 8\times 10^{-3}$. Notice that \texttt{ULA} is unable to produce samples from the narrow region of this distribution.

We plot the KSD for four of the systems in Figure \ref{fig:funnelksd}. We do not report the KSD for the \texttt{ULA} system as it does not produce samples from the narrow region of the posterior and leads to poor estimates of the KSD. Notice that all the systems with the map-defined metric outperform the metric defined in \cite{girolami2011riemann}, and that the \texttt{TMULA} discretization outperforms the direct Euler-Maruyama discretization of \texttt{TMRMLD} (EMRMLD). These findings are further supported by Table \ref{table:hybridrosen} and Figure \ref{fig:funnelmse}, which show the asymptotic variance and MSE convergence plots for test functions $\phi_1(\mu,\gamma) = \exp(\gamma)$, $\phi_2(\mu,\gamma) = \gamma + \mu$, and $\phi_3(\mu,\gamma) = \gamma^2 + \mu^2$.

  \begin{figure}[h]
    \centering
\includegraphics[width = 0.6\textwidth]{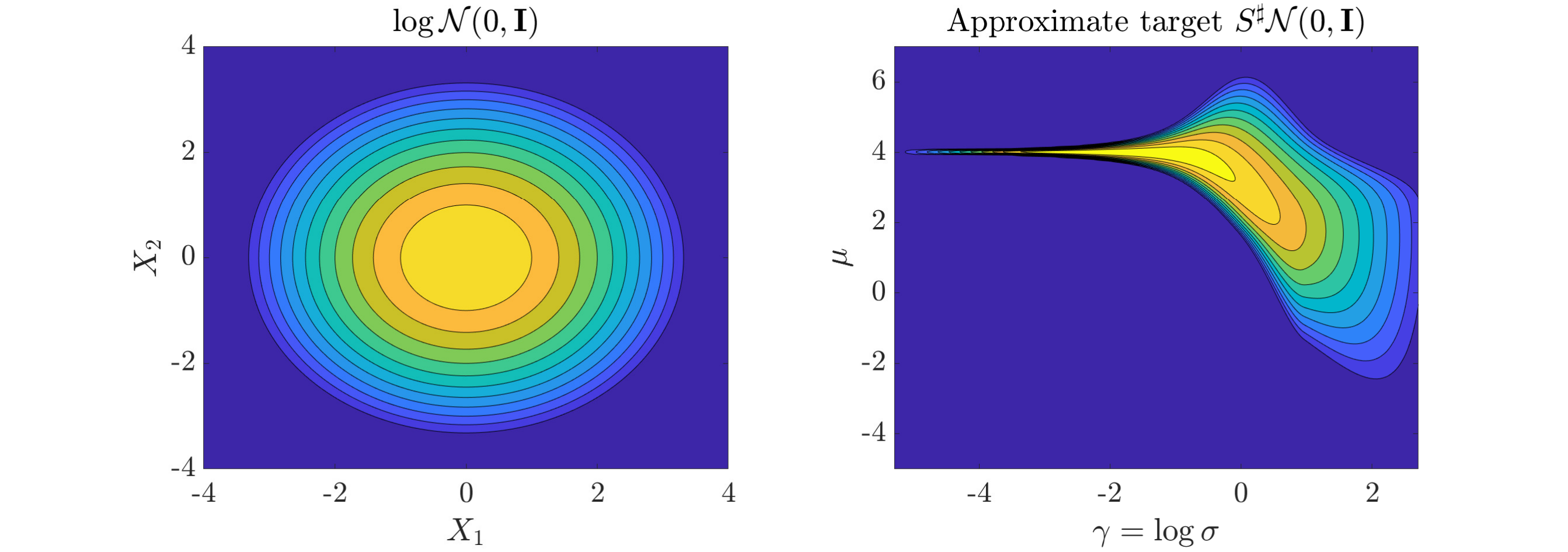}
    \includegraphics[width = 0.6\textwidth]{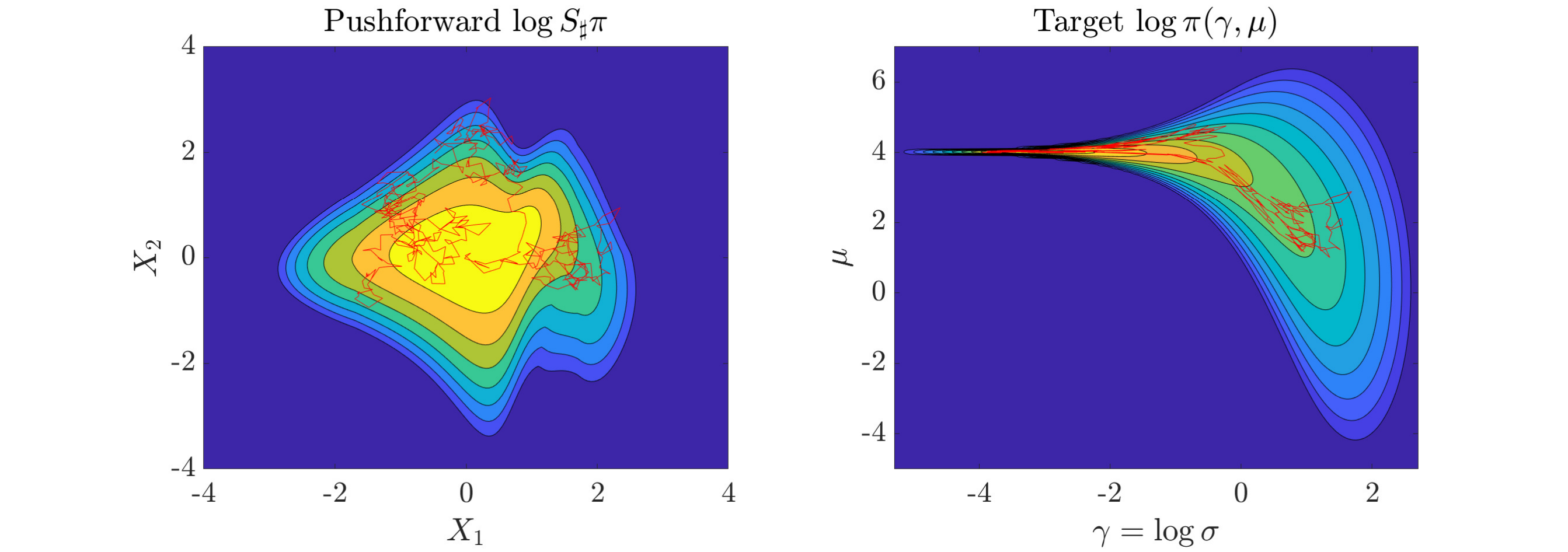}
    \caption{Posterior for parameters of a normal distribution example. Top right figure shows the pullback of normal distribution through map $S$. Figures in the bottom row show the pushforward of the target $\pi$ through map $S$ and the true distribution $\pi$. Red curves show the Langevin trajectory on the pushforward distribution and on the target distribution. }
    \label{fig:funneldensities}
  \end{figure}

  \begin{figure}
    \centering
    \includegraphics [width = \textwidth]
    {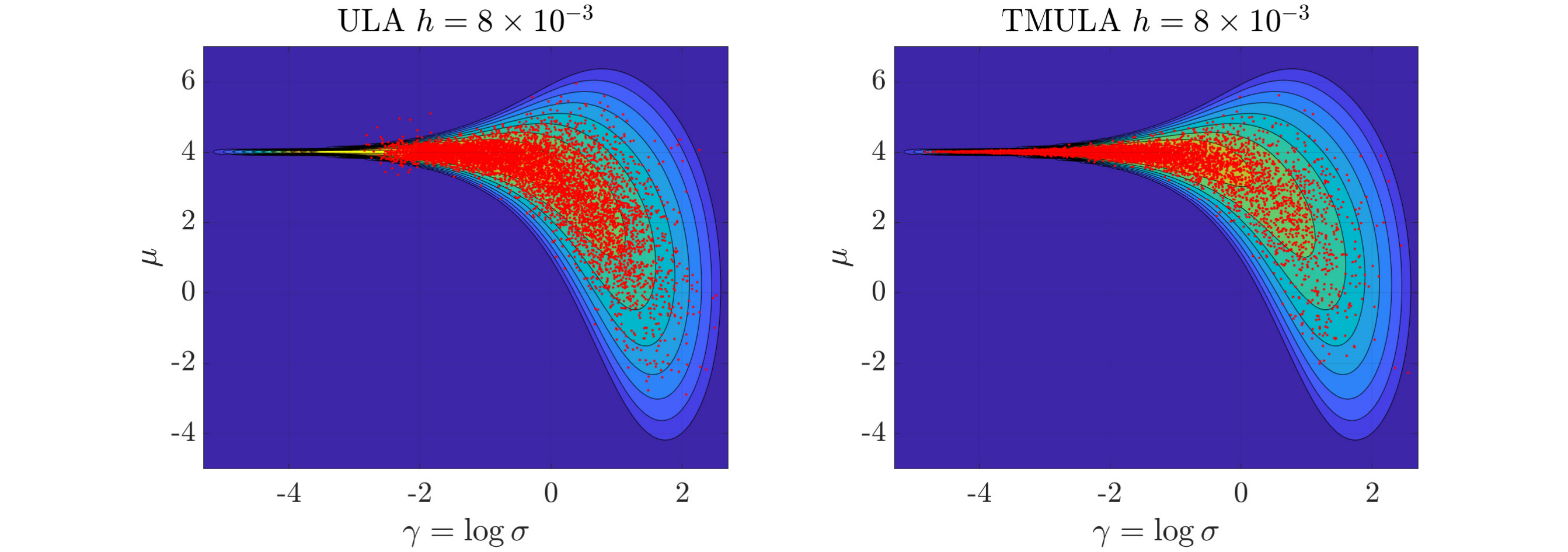}
    \caption{Posterior for parameters of a normal distribution. Left figure shows the samples produced by ULA, right figure shows produced by TMULA. For this step size, note that ULA is unable to produce samples from the narrow region of the distribution.  }
    \label{fig:ulavstmula}
  \end{figure}

  \begin{figure}
    \centering
    \includegraphics[width = 0.5\textwidth]{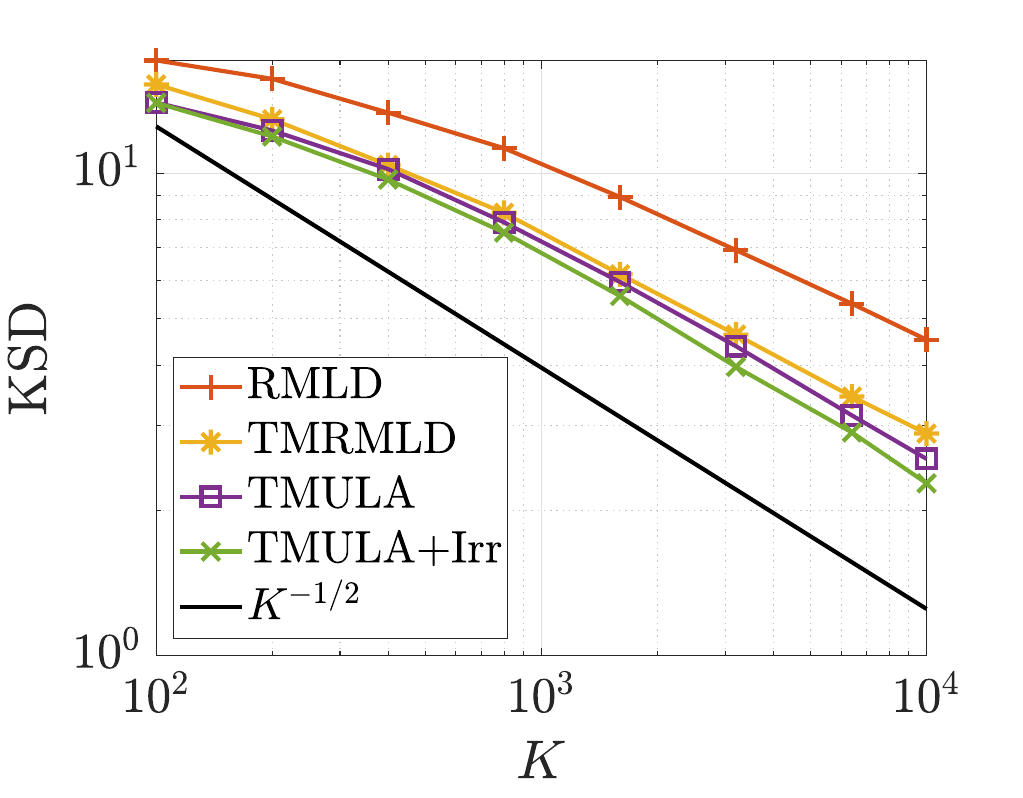}
    \caption{KSD for the funnel distribution. Note that we do not plot the KSD of the \texttt{ULA} process as it is unable to produce samples from the narrow region of the posterior and will produce poor estimates of the KSD. }
    \label{fig:funnelksd}
  \end{figure}

  \begin{figure}
    \centering
\includegraphics[width = \textwidth]{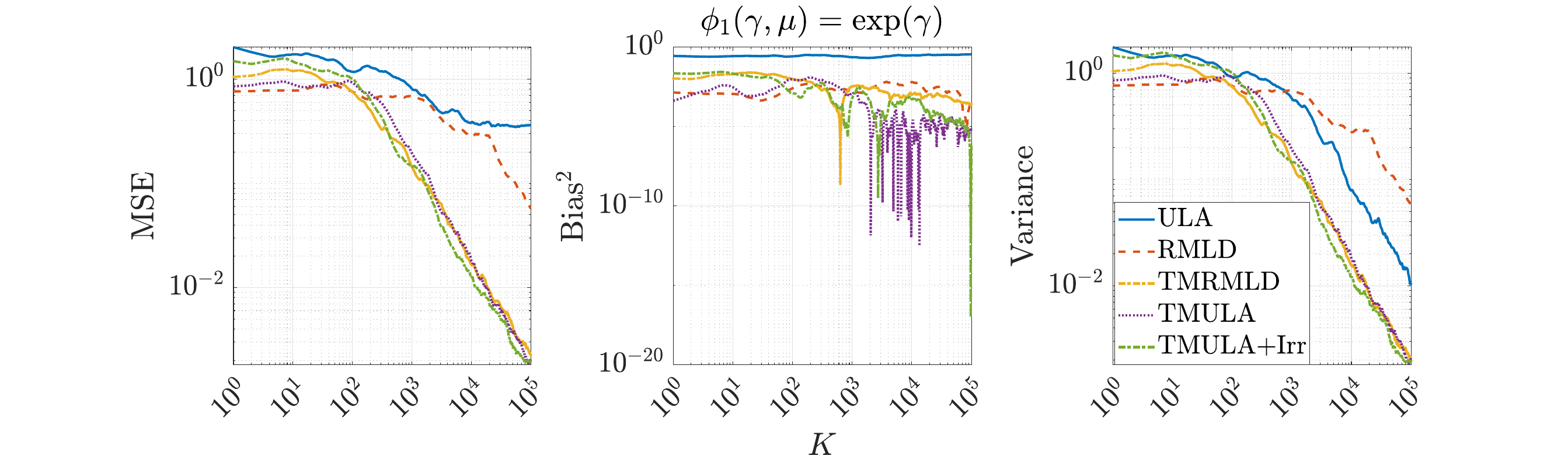}
    \includegraphics[width = \textwidth]{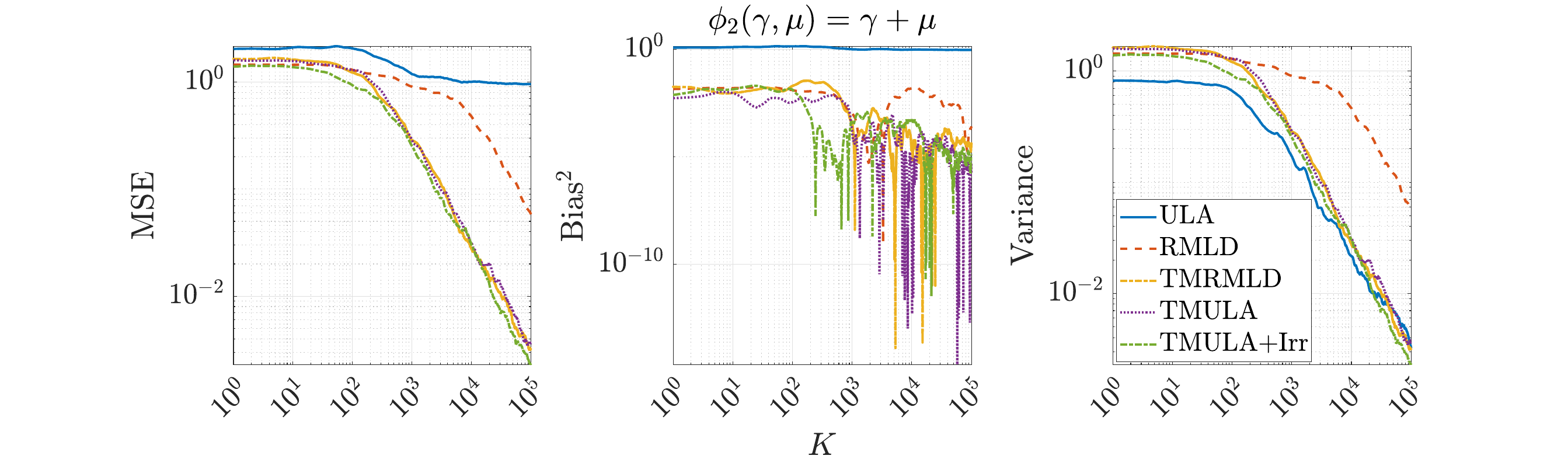}
    \includegraphics[width = \textwidth]{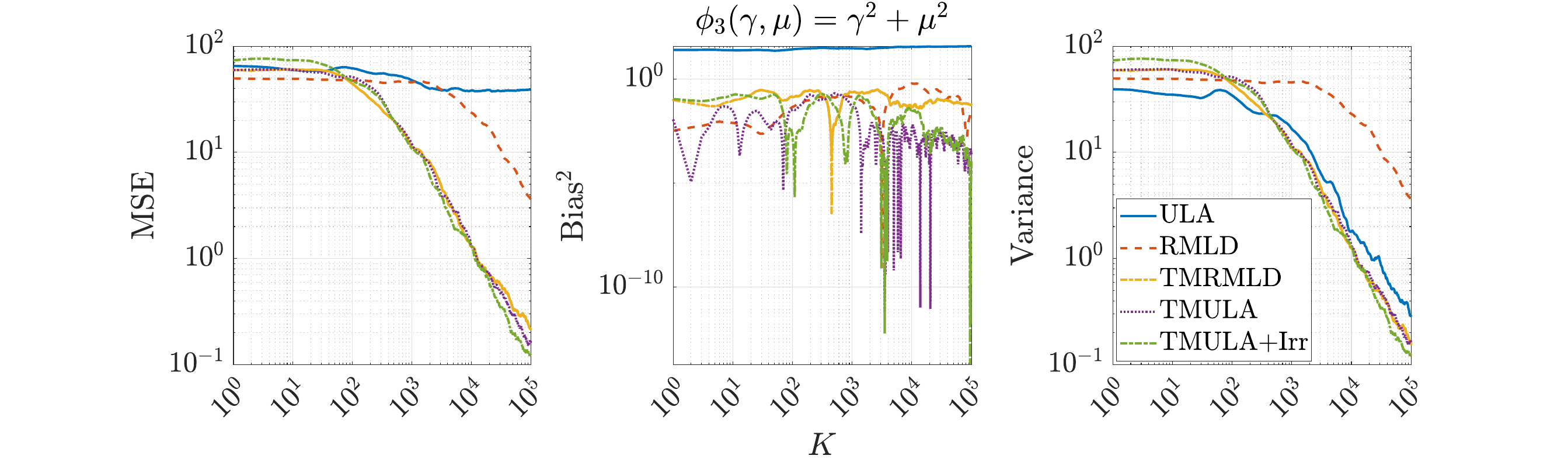}
    \caption{MSE, bias, and variance for the funnel distribution.}
    \label{fig:funnelmse}
  \end{figure}

  \begin{table}[H]
    \centering
    {\footnotesize
    \begin{tabular}{|l|l|l|l|l|l|l|}
    \hline
                           & $\Ex[\text{AVar}_{\phi_1}]$    & $\text{Std}[\text{AVar}_{\phi_1}]$ & $\Ex[\text{AVar}_{\phi_2}]$    & $\text{Std}[\text{AVar}_{\phi_2}]$  & $\Ex[\text{AVar}_{\phi_3}]$    & $\text{Std}[\text{AVar}_{\phi_3}]$ \\ \hline
    \texttt{ULA}                    &8.759 & 1.797 &1.957 & 0.4774 &195.4 &35.30  \\ \hline
    \texttt{RMLD}                    &25.46 &7.550 & 28.82 & 2.860   &1558 &184.8   \\ \hline
    \texttt{TMRMLD}       & 1.344 &$\mathbf{0.2057}$ & 2.655 & 0.3705             & 108.7 & 15.48  \\ \hline
    \texttt{TMULA}  &1.444 &0.2061 & 2.480 & 0.3475   & 114.8  & 14.00   \\ \hline
    \texttt{TMULA + Irr} \hspace{-8pt}  &$\mathbf{ 1.243}$ & 0.2131 & $\mathbf{1.961}$ & $\mathbf{0.2851}$           & $\mathbf{92.72}$ & $\mathbf{12.89}$     \\ \hline

    \end{tabular}
    \caption{Asymptotic variance estimates for the funnel distribution. } \label{table:funnel}
    }

  \end{table}

\subsection{Hybrid Rosenbrock distribution}\label{SS:RosenbrockDistribution}

We demonstrate our methodology on a seven-dimensional hybrid Rosenbrock distribution. This distribution was originally constructed as a generalization of the 2D Rosenbrock distribution for the purposes of testing sampling schemes \cite{pagani2022n}. The density is defined as
\begin{align*}
 \pi(\textbf{y}) \propto \exp\left\{-a(y_1-\mu)^2 - \sum_{j = 1}^{n_2} \sum_{i = 2}^{n_1} b_{ji} (y_{j,i} - y_{j,i-1}^2)^2 \right\}
\end{align*}
where $\mu, y_{j,i} \in \R$, $a,b_{j,i}$ are strictly positive, and the dimension of the problem is $d = (n_1 - 1)n_2 +1$. We choose $n_1 = 4$, $n_2 = 2$, $\mu = 1$, $a = 30$, and $b_{j,i} = 20$ for all $(j,i)$.
It is important to note that the discrete-time stochastic process that ULA produces for this density is transient, as the drift of the resulting SDE has a cubic term. This phenomenon is well-documented and so ULA cannot be used to reliably produce samples \cite{casella2011stability}. While a perfect transport map is able to alleviate this problem, an approximate map may not be able to eliminate all the higher order terms, which may cause the stochastic process generated by TMULA to be transient. To solve this problem, for this example we use the split-step implicit Langevin algorithm, which is an implicit algorithm. A deterministic implicit Euler step is first performed before noise from the discretized Brownian motion is added. See \cite{mattingly2002ergodicity} for details about this method. In Equation \ref{eq:splitstep} we show how we incorporate the implicit scheme with \texttt{TMULA}. We refer to this resulting implicit scheme as TMUILA.

  \begin{align}
  \begin{dcases}
    S(Y^*) = S(Y^k) + h \mathbf{J}^\top_S(Y^*) ^{-1}\left[\nabla_Y \log \pi(Y^*) - \sum_{i = 1}^d \left(\frac{\partial S_i}{\partial y_i}(Y^*) \right)^{-1} H_i(Y^*) \right] \\
    X_{k+1} = S(Y^*) + \sqrt{2h} \xi^{k+1} \\
    Y_{k+1} = T(X_{k+1}),
  \end{dcases} \label{eq:splitstep}
\end{align}
where $H_i(Y^k) = \left[\frac{\partial^2 S_i}{\partial y_1 \partial y_i}, \cdots,  \frac{\partial ^2S_i}{\partial y_d\partial y_i} \right]^\top$,
where $\xi^{k+1}\sim \mathcal{N}(0,\mathbf{I})$.

A transport map is learned with 2500 training samples. In this example, the training samples can be obtained directly since the true map that normalizes the hybrid Rosenbrock distribution is known. We plot the training samples in left of Figure \ref{fig:rosentrain}. The function $f$ in the transport map is optimized over the span of Hermite polynomials with total order 2 that is then rectified so that its Jacobian is always positive definite.

We compare implicit unadjusted Langevin (UILA) with the implicit version of TMULA (TMUILA) both with step size $h = 0.01$. In right of  Figure %\ref{fig:ksdrosen}
\ref{fig:rosentrain}, we compare the KSD of the two systems. Notice that the KSD of UILA plateaus quickly, which implies that the samples are quite biased while TMUILA does not exhibit this phenomenon in our experiments. This implies that the use of the transport map allows us to use larger step sizes when simulating TMUILA. These results are further supported by the asymptotic variance estimates in Table \ref{table:funnel} and the MSE plots in Figure \ref{fig:rosenmse} test functions $\phi_1(Y) = \sum_{ i = 1}^7 Y^i$ and $\phi_s(Y) = \sum_{i = 1}^7  (Y^i)^2$.

\begin{figure}[h]
  \centering
  \includegraphics[width = 0.45\textwidth]{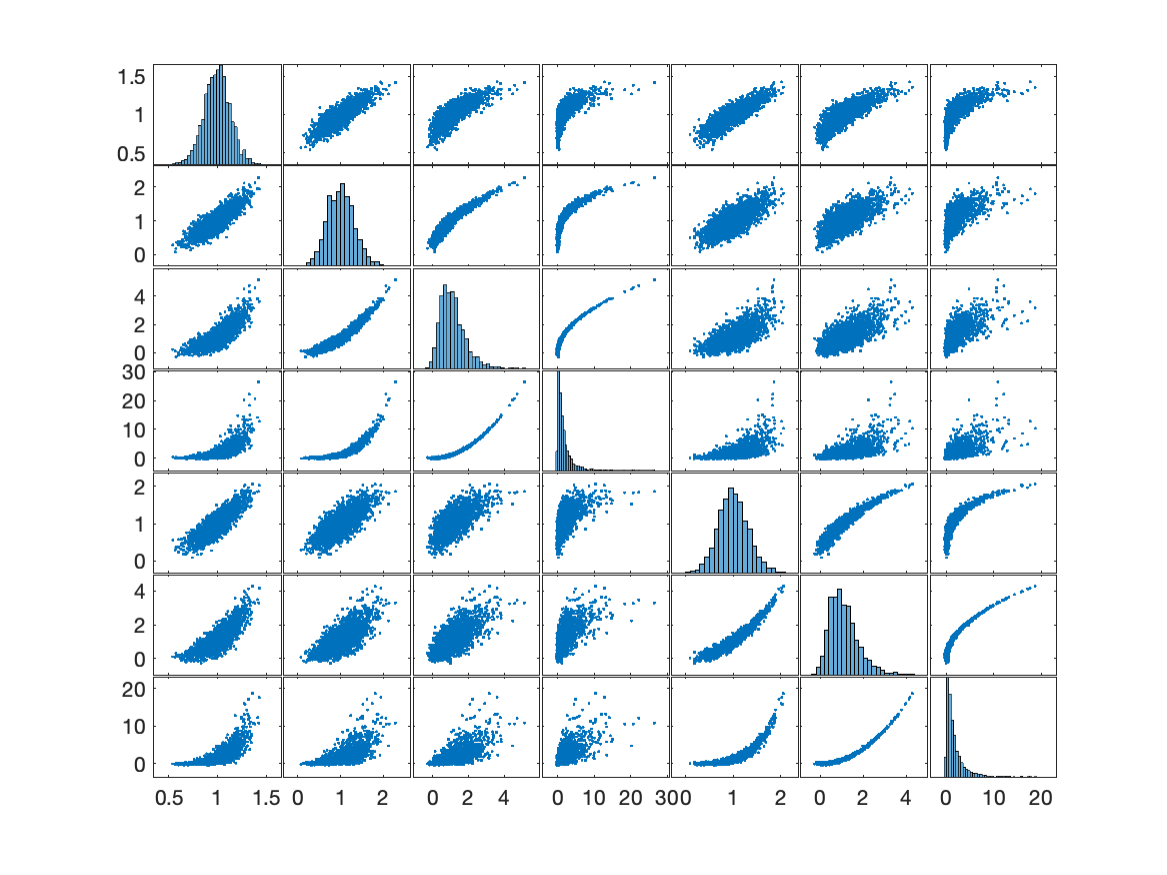}
  \includegraphics[width = 0.45\textwidth]{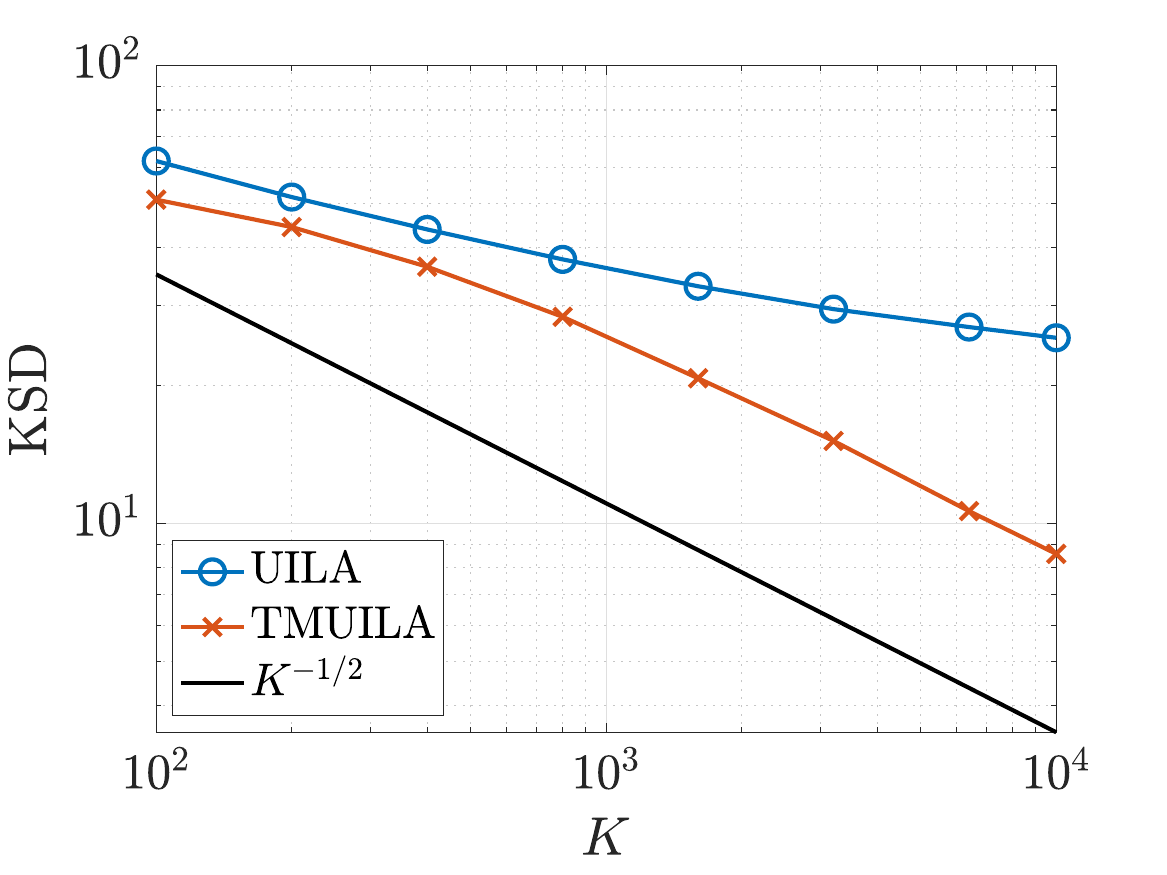}
  \caption{Left: Training samples for the hybrid Rosenbrock example. Right: Kernelized stein discrepancy for the hybrid Rosenbrock distribution.} \label{fig:rosentrain}
\end{figure}

\begin{figure}[h]
  \centering
  \includegraphics[width = \textwidth]{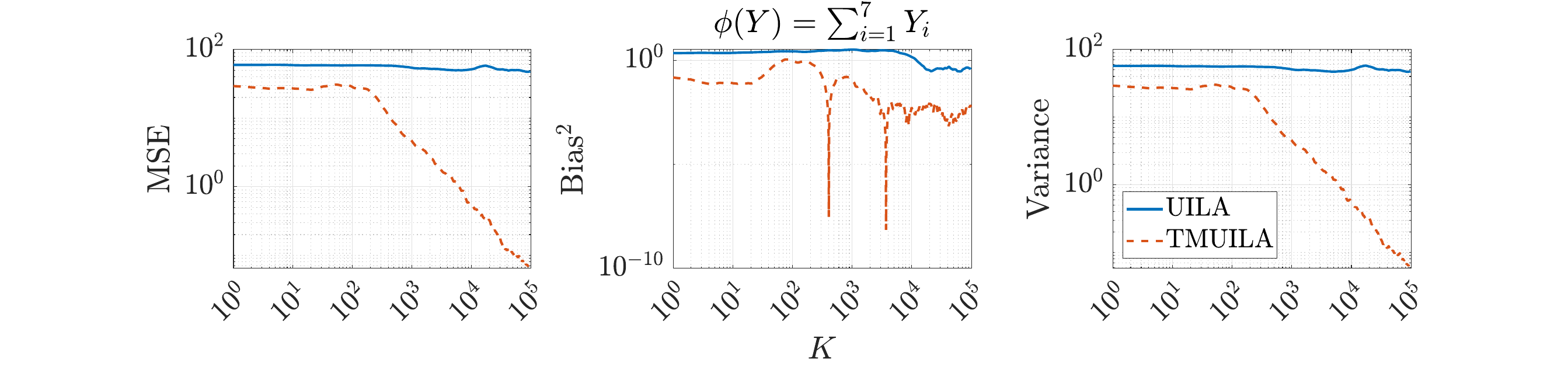}
  \includegraphics[width = \textwidth]{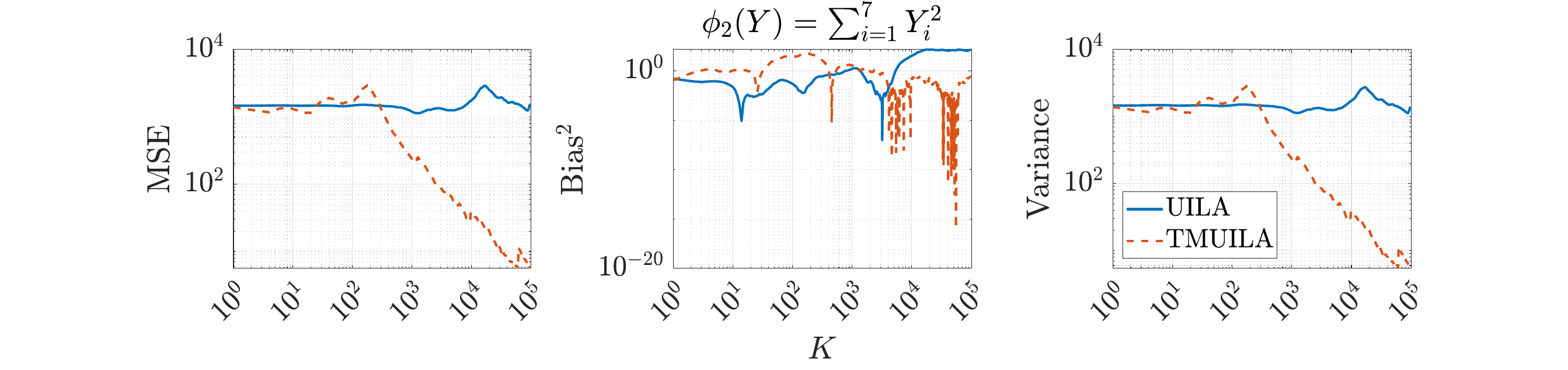}
  \caption{MSE for hybrid Rosenbrock distribution.} \label{fig:rosenmse}
\end{figure}

\begin{table}[H]
  \centering
  \begin{tabular}{|l|l|l|l|l|}
  \hline
                         & $\Ex[\text{AVar}_{\phi_1}]$    & $\text{Std}[\text{AVar}_{\phi_1}]$ & $\Ex[\text{AVar}_{\phi_2}]$    & $\text{Std}[\text{AVar}_{\phi_2}]$  \\ \hline
  \texttt{UILA}                    &6762 & 2663 &$6.957\times10^6$ &$5.185\times 10^6$   \\ \hline
  \texttt{TMUILA}  &$\mathbf{65.03}$ &$\mathbf{28.54}$ &$\mathbf{6506}$ & $\mathbf{1284} $      \\ \hline
  \end{tabular}
  \caption{Asymptotic variance estimates for the hybrid Rosenbrock distribution. }
  \label{table:hybridrosen}
  \end{table}

\subsection{Multimodal distributions}\label{SS:MultimodalDistribution}

We consider a multimodal example. Efficient Langevin sampling of not strongly log-concave distributions is an active area of research. In contrast to the previous examples, we only take note of some qualitative features and challenges of the method for sampling from distributions with multiple modes. The target distribution is a mixture of four Gaussians with means located at $(-4,-4),(4,-4),(-4,4),(4,4)$ all with identity covariance matrices. The weights for each Gaussian are $w_1 = 0.337, w_2 = 0.050, w_3 = 0.284, w_4 = 0.328$. In Figure \ref{fig:multimodaltarget} we plot the training samples and the density of the target.% distribution.

Two maps are learned adaptively with \texttt{ATM} using $N = 200$ and $N = 2000$ samples \cite{baptista2022atm, baptista2020adaptive}. When applying TMULA with the learned maps, we often observed that the method would still get stuck in certain modes of the target distribution despite the fact that the learned map would push the samples to what appears to be a unimodal distribution. In Figure \ref{fig:multimodalpushforward}, we show the pushforward samples through the learned maps with $N = 200$ and $N = 2000$ training samples. Notice that while the scatterplot of samples pushed forward through the learned maps looks quite normal and unimodal, the pushforward \textit{densities} have separatrices that make it difficult for TMULA to transition between the multiple modes. More samples can reduce these separations; the $N =2000$ case shows reduced intensity of the separatrices.

We conjecture that this phenomenon further demonstrates the difficulty of sampling from multimodal distributions with Langevin dynamics. While the transport map can normalize the shape of the distribution, it has challenges in resolving the transitions between modes. We note that this phenomenon is very sensitive to the errors of the map. In Figure \ref{fig:multimodalmaps}, we plot the exact and learned maps (with $N = 200$). Notice that the learned map approximates the true normalizing map quite well, even though it still results in a pushforward density that is difficult to sample using Langevin dynamics. We conjecture that separatrices follow from the fact that, when learning maps from samples, it is intrinsically challenging to capture the correct slope (or the  Jacobian determinant) of $S$ in regions of low density $\pi$, because there is necessarily a lack of sample information in such regions. Further research into this phenomenon and techniques for learning maps that avoid these separatrices are ongoing. Our focus here is on how learned maps interact with ULA sampling, and we leave further investigation of how best to obtain maps for multi-modal targets to other work.

\begin{figure}[h!]
  \centering
  \includegraphics[width = 0.75\textwidth]{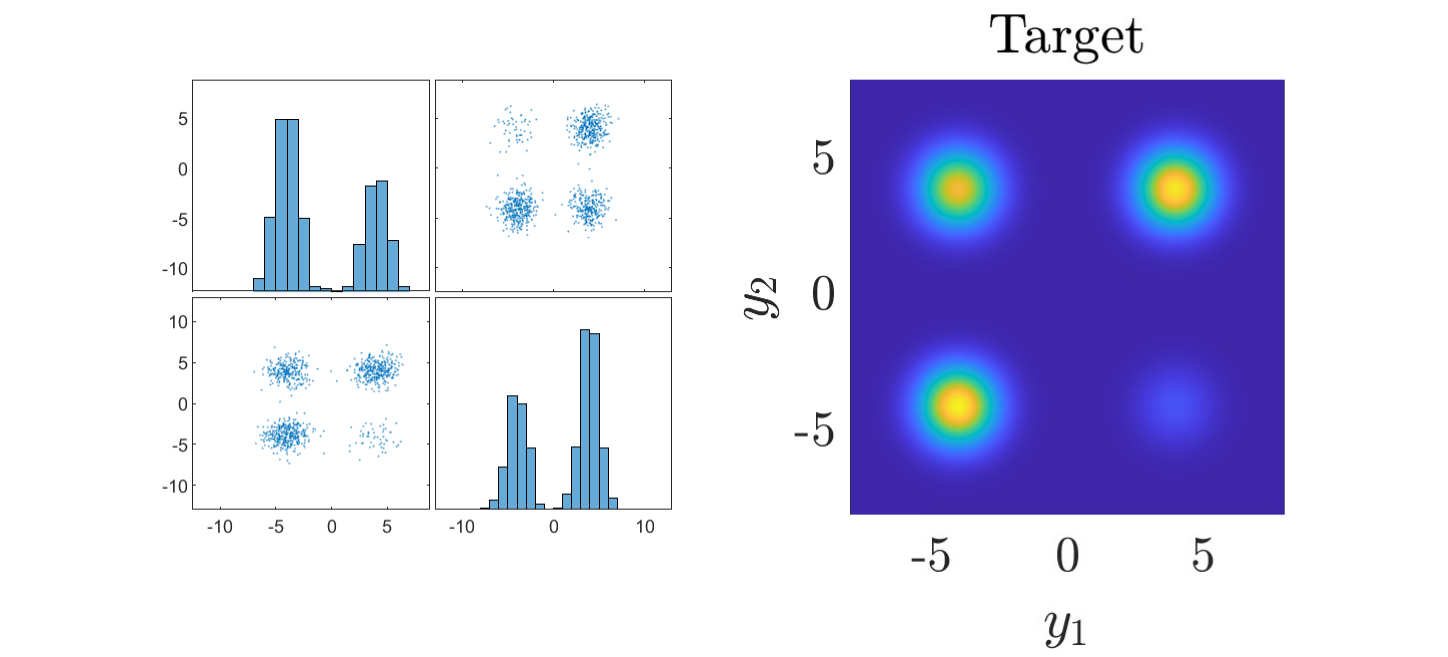}
  \caption{Gaussian mixture target distribution. Left figure shows samples, right figure shows target density. }
  \label{fig:multimodaltarget}
\end{figure}
\begin{figure}[h!]
  \centering
 \subfloat[$S$ is learned from $N = 200$ samples]{
  \includegraphics[width = 0.49\textwidth]{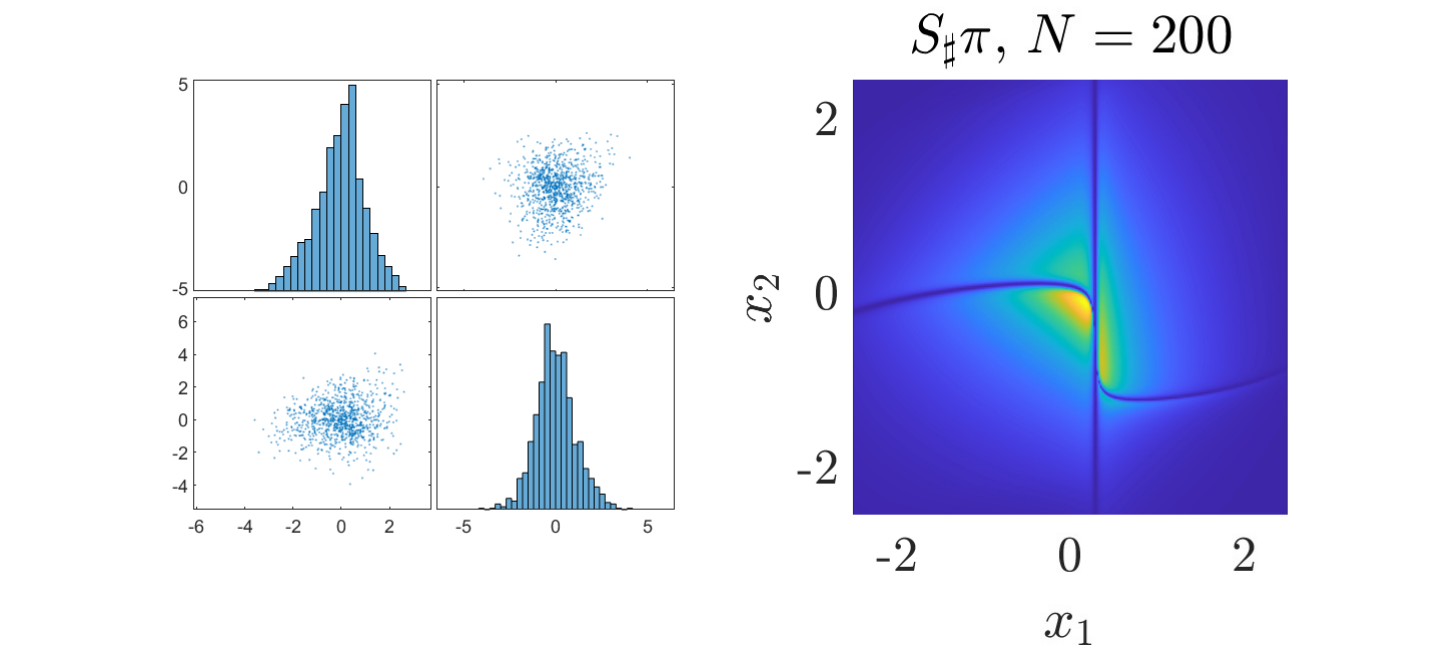}
 }
\subfloat[$S$ is learned from $N = 2000$ samples]{
\includegraphics[width = 0.49\textwidth]{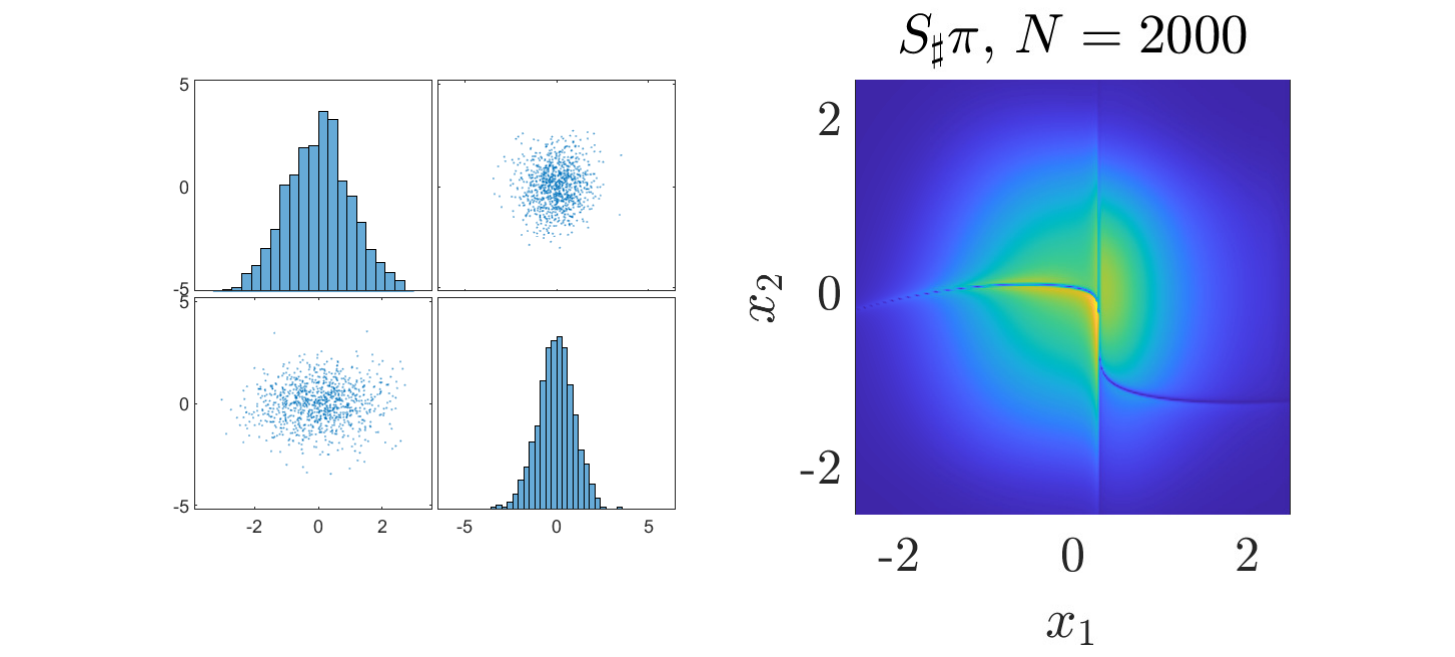}
}
  \caption{Mixture of Gaussians. Notice that while the pushforward samples in each figure look quite normal, the pushforward density $S_\sharp \pi$ (computed exactly using the change-of-variables formula) shows boundaries separating the Gaussian into four regions.   }
  \label{fig:multimodalpushforward}
\end{figure}

\begin{figure}[h!]
  \centering
  \includegraphics[width = \textwidth]{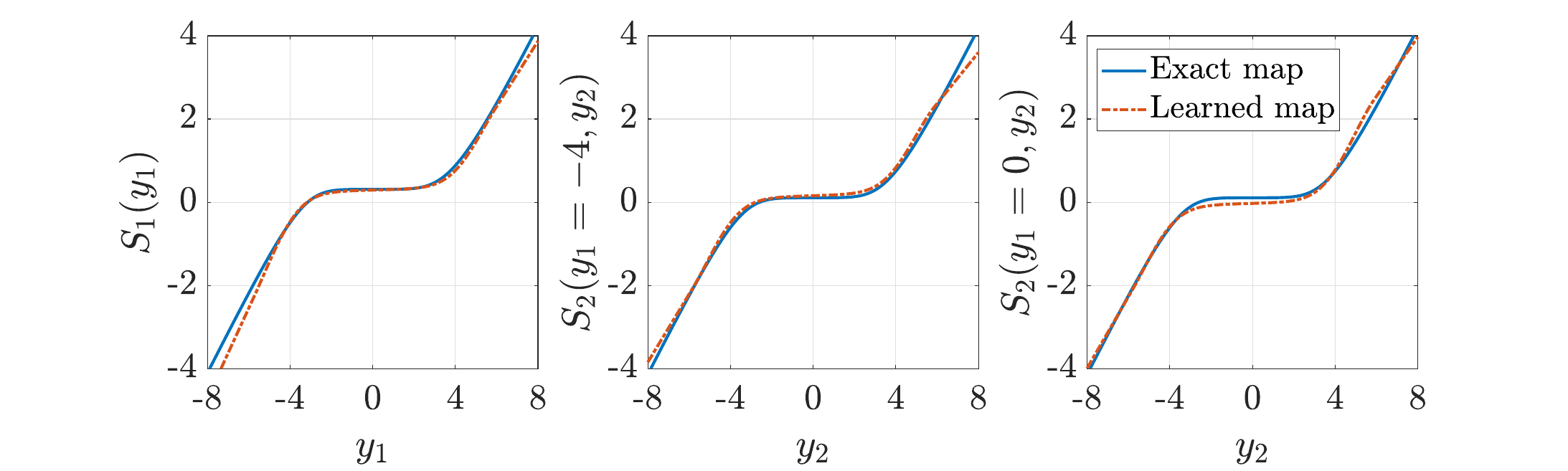}
  \caption{Comparison showing that a transport map can be well approximated but still lead to a problematic pushforward density. Here, $N = 200$. }
  \label{fig:multimodalmaps}
\end{figure}

\section{Outlook}\label{S:Conclusion}

\textcolor{black}{In this paper we studied the effects of transport maps on the unadjusted Langevin algorithm. We discussed how the properties of transport maps and target distributions affect the resulting TMULA algorithm. The general question of how to judiciously construct transport maps to optimize sampling performance is left for future work. We discuss many open and interesting questions that work towards maps for optimally improving sampling. }

{First, building on the results of Section~\ref{subsec:tmulaanalysis}, it would be useful to improve our theoretical understanding of how to characterize the transport map within a given approximate class (i.e., with finite expressivity) that maximizes the efficiency of TMULA sampling, and to develop practical training objectives for learning such a map.}
Second, in the numerical example of Section \ref{SS:MultimodalDistribution} we demonstrated the issues that approximate transport maps face in the presence of multimodality. It would be useful to pinpoint the essence of this challenge and to devise workarounds.
Third, a transport can sometimes produce a pushforward density with tails lighter than Gaussian, such that the resulting discretized Langevin process is transient. It would be desirable to understand this phenomenon better.

\appendix
\section{Proofs for Section \ref{subsec:connections}}
\label{app:connections}
\begin{proof}[Proof of Theorem \ref{prop:TMRMLD}]
  We first derive the SDE of the diffusion process $Z(t) = T(X(t))$. By It\^o's lemma \cite{oksendal2013stochastic}, the $k$--th component of $Z(t)$ can be written as

  \begin{align}
    \de Z(t) = \J_T \nabla_X \log \eta(S(Z(t))) \de t + c(Z(t)) \de t + \sqrt{2} \J_T \de W(t), \label{eq:ito}
  \end{align}
  where
   $ c_k(Z) = \sum_{i = 1}^d\frac{\partial^2 T_k}{\partial x_i^2} = \sum_{i = 1}^d\sum_{j = 1}^d \frac{\partial ^2 T_k}{\partial y_j \partial x_i} \cdot \frac{\partial T_j}{\partial x_i}.$
  The second equality is true by the multivariate chain rule and will be useful later.

  To construct the reversibly perturbed OLD on $\pi$, we apply the reversible perturbation to $\{Y(t)\}$ via the matrix $\B = (\J_S^{\top} \J_S)^{-1}$. After substituting $\pi(y) = \eta(S(y)) \det \J_S(y)$, we have
  \begin{align}
    \de Y(t) =(\J_S^{\top} \J_S)^{-1} \nabla_y \log \left[ \eta(S(Y(t)))\det(\J_S)\right] \de t + \nabla_y \cdot (\J_S^{\top} \J_S)^{-1} \de t + \sqrt{2} \J_S^{-1} \de W(t).
    \label{eq:rmld}
   \end{align}
   First note that the diffusion term is equivalent to that of Equation \eqref{eq:ito}, so we only need to compare the drift terms. Next, notice that %since $\partial S_i/\partial y_k = 0$ for $i<k$,
\begin{align*}
  \frac{\partial}{\partial y_k} \log (\eta(S(y))) &= \sum_{i = 1}^d \frac{\partial S_i}{\partial y_k} \frac{\partial}{\partial x_i}\log(\eta(S(y))).
\end{align*}
Therefore, $\nabla_y \log \eta(S(y)) = \J_S^\top \nabla_x \log \eta(S(y))$ and \begin{align*} (\J_S^{\top} \J_S)^{-1} \nabla_y \log \eta(S(Y)) = \J_S^{-1} \nabla_x\log{\eta(S(y))}\end{align*}. This exactly matches the first part of the drift term in Equation \eqref{eq:ito}.

For the divergence term, first note the following identity. Let $\A$ and $\D$ be $d\times d$ matrix-valued functions. Then observe that
\begin{align*}
  (\nabla \cdot \A \D)_k &= \sum_{j= 1}^d \frac{\partial}{\partial y_j}(\A\D)_{kj} = \sum_{j = 1}^d \sum_{i = 1}^d \frac{\partial}{\partial y_j} \A_{ki}\D_{ij}\\
  & = \sum_{i = 1}^d \sum_{j = 1}^d \frac{\partial \A_{ki}}{\partial y_j} \D_{ij} + \A_{ki}\frac{\partial \D_{ij}}{\partial y_j}   = \sum_{i = 1}^d \sum_{j = 1}^d\left[ \frac{\partial \A_{ki}}{\partial y_j} \D_{ij}\right] + (\A \nabla \cdot \D)_k.
\end{align*}
Applying this identity to $\B= (\J_S^{\top} \J_S)^{-1} = \J_T\J_T^\top$, we obtain
\begin{align*}
  (\nabla_y \cdot \J_T \J_T^\top)_k &= \sum_{i,j = 1}^d \left[\frac{\partial (\J_T)_{ki}}{\partial y_j} (\J_T)^\top_{ij} \right] + (\J_T \nabla_y \cdot \J_T^\top)_k \\ &= \sum_{i,j = 1}^d \left[\frac{\partial^2 T_k}{\partial y_j\partial x_i} \cdot \frac{\partial T_j}{\partial x_i}\right] + (\J_T \nabla_y \cdot\J_T^\top)_k.
\end{align*}
Notice that the first term is exactly the It\^o correction term. Hence, we only need to show that $(\J_S^{\top} \J_S)^{-1}\nabla_y \log \det \J_S + \J_T \nabla_y \cdot \J_T^\top = 0$. To avoid issues with the chain rule relating $Y$ and $X$, it suffices to show that $\nabla_y \log\det \J_S + \J_S^\top \nabla_y \cdot (\J_S^\top)^{-1} = 0$.

For the first term, observe that
\begin{align*}
  \frac{\partial}{\partial y_k}\log \det \J_S  &= \sum_{ i,j = 1}^d  ((\mathbf{J}_S^\top)^{-1})_{ij} \left(\frac{\partial \mathbf{J}_S }{\partial y_k} \right)_{ij}  = \sum_{i,j = 1}^d \left(\frac{\partial^2 S_i}{\partial y_k \partial y_j} \right) ((\mathbf{J}_S^\top)^{-1})_{ij} \\
  & = \sum_{i,j = 1}^d \left(\frac{\partial \mathbf{J}_S }{\partial y_j} \right)_{ik} ((\mathbf{J}_S^\top)^{-1})_{ij}
\end{align*}

Letting $:j$ denote the $j$--th column vector, we see
\begin{align*}
  (\J_S^\top \nabla \cdot (\J_S^\top)^{-1})_k &= -\left(\J_S^\top \sum_{j = 1}^d \left[(\J_S^\top)^{-1} \frac{\partial \J_S^\top}{\partial y_j} (\J_S^\top)^{-1} \right]_{:j} \right)_k\\
  &= -\left(\sum_{j = 1}^d \left[\frac{\partial \J_S^\top}{\partial y_j} (\J_S^\top)^{-1} \right]_{:j} \right)_k \\
  & = - \sum_{j = 1}^d \left[\frac{\partial \J_S^\top}{\partial y_j} (\J_S^\top)^{-1} \right]_{kj} \\ & = -\sum_{j,i = 1}^d \left(\frac{\partial\J_S^\top}{\partial y_j} \right)_{ki} ((\J_S^\top)^{-1})_{ij}.
\end{align*}
which exactly the negative of the derivative of the log determinant of $\J_S$. We therefore conclude that Equation \eqref{eq:ito} exactly matches \eqref{eq:rmld}.
\end{proof}

\begin{proof}[Proof of Theorem \ref{prop:tmgiirr}]
  Apply It\^o's formula to get
  \begin{align*}
    \de Y_k(t) = \sum_{i = 1}^d \frac{\partial T_k}{\partial x_i }\hspace{-2pt}\left[\frac{\partial}{\partial x_i}\log \eta(X(t)) + \left(\mathbf{D}\nabla \log \eta(X(t)) \right)_i \right]\hspace{-2pt} \de t &+ \sum_{i= 1}^d \frac{\partial^2 T_k}{\partial x_i^2 } \de t \\&+ \sqrt{2}\sum_{i = 1}^d \frac{\partial T_k}{\partial x_i} \de W_i(t).
  \end{align*}
  The first, third, and fourth terms are identical to the ones in Theorem \ref{prop:TMRMLD} when irreversibility is not considered. Therefore, we only need to address the second term. Note that
  \begin{align*}
    \mathbf{J}_T\mathbf{D} \nabla_x \log \eta(x)  &= \J_T\mathbf{D}\mathbf{J}_T^\top \nabla_y \log \eta(S(y)) \\
    & = \mathbf{J}_T\mathbf{D}\mathbf{J}_T^\top \nabla_y \log \left( \eta(S(y)) \det \mathbf{J}_S(y) \right) - \mathbf{J}_T\mathbf{D}\mathbf{J}_T^\top\nabla_y \log \det \mathbf{J}_S(y) \\
    & = \mathbf{J}_T\mathbf{D}\J_T^\top \nabla_y \log \pi(y) - \mathbf{J}_T\mathbf{D}\mathbf{J}_T^\top \nabla_y\log \det \mathbf{J}_S(y).
  \end{align*}
  We only need to show that the last term is equal to $\nabla_y \cdot \mathbf{J}_T\mathbf{D}\mathbf{J}_T^\top$.  From the proof of the previous Theorem, first observe that
  \begin{align}
    (\nabla_y \cdot \mathbf{J}_T\mathbf{D}\mathbf{J}_T^\top)_k = \sum_{i,j = 1}^d \left[ \frac{\partial}{\partial y_j}(\mathbf{J}_T\mathbf{D})_{ki} \frac{\partial T_j}{\partial x_i}\right] + (\mathbf{J}_T\mathbf{D} \nabla_y \cdot \mathbf{J}_T^\top)_k \label{eq:identity}
  \end{align}
  and that $\mathbf{J}_T \mathbf{D} \nabla_y \cdot \mathbf{J}_T^\top = -\mathbf{J}_T \mathbf{D}\mathbf{J}_T^\top \nabla_y \log \det \mathbf{J}_S$. Therefore, we need to show that the first term on the right hand side in \eqref{eq:identity} above is identically zero. We compute
  \begin{align*}
    \sum_{i,j = 1}^d \left[ \frac{\partial}{\partial y_j}(\mathbf{J}_T\mathbf{D})_{ki} \frac{\partial T_j}{\partial x_i}\right] &= \sum_{i,j,l = 1}^d \left[ \mathbf{D}_{li}\frac{\partial^2 T_k}{\partial y_j\partial x_l} \frac{\partial T_j}{\partial x_i}\right]  = \sum_{i,l = 1}^d \left[\mathbf{D}_{li} \frac{\partial^2 T_k}{\partial x_i\partial x_l} \right] \\
    & = \sum_{i,l =1, i>l}^d \mathbf{D}_{li}\left[\frac{\partial^2 T_k}{\partial x_i\partial x_l}-\frac{\partial^2 T_k}{\partial x_l\partial x_i} \right] = 0.
  \end{align*}

\end{proof}

\section{Proofs for Section \ref{subsec:tmulaanalysis}}
\label{app:wasser}

The proof is straightforward: we use the same approach as taken in \cite{hsieh2018mirrored} and relate the convergence of ULA on $S_\sharp \pi$ to convergence on $\pi$ by relating the Wasserstein distance between $\eta^k$ and $\eta$ with that of $\pi^k$ and $\pi$ \cite{durmus2019high}.
\begin{lemma}
  Let $\eta^k$ and $\pi^k$ be the distribution of the discrete-time process $X^k$ and $Y^k$ at time step $k$, respectively.
  \begin{align}
    \mathcal{W}^2_2(\pi^k,\pi) \le \frac{1}{\rho^2} \mathcal{W}^2_2(\eta^k,\eta).
  \end{align}
\label{lemma:wasser}
\end{lemma}
\begin{proof}
  We use the fact that $S$ is $\rho$-strongly monotone:
  \begin{align*}
    \mathcal{W}_2^2(\eta^k,\eta) = \mathcal{W}_2^2(S_\sharp \pi^k,S_\sharp \pi) &= \inf_W \int \|x - W(x) \|^2 \de S_\sharp \pi^k(x)  \\
    & = \inf_W \int \|S(x) - W\circ S(x) \|^2 \de \pi^k(x) \\
    &\ge \rho^2 \inf_W \int \|x - S^{-1} \circ W \circ S(x) \|^2 \de \pi^k(x).
  \end{align*}
  Since $W$ is such that $W_\sharp S_\sharp \pi^k = S_\sharp \pi$, this implies that $S_\sharp^{-1} W_\sharp S_\sharp \pi^k = \pi$, and we therefore have
  \begin{align*}
    \mathcal{W}_2^2(\pi^k,\pi) \le \frac{1}{\rho^2}\mathcal{W}_2^2(\eta^k,\eta).
  \end{align*}
\end{proof}

\begin{proof}[Proof of Theorem \ref{prop:wasser}]
 We apply Theorem 5 of \cite{durmus2019high} to ULA on $\eta$, and then use Lemma \ref{lemma:wasser} to derive a bound for $\pi$. Observe that

  \begin{align*}
    \mathcal{W}_2^2(\pi^k,\pi) &\le \frac{1}{\rho^2} \mathcal{W}_2^2(\eta^k,\eta) \le  \frac{1}{\rho^2}  \left[ \left(1- \frac{\kappa h}{2}\right)^k \left(2\|y-y^\star\|^2 + \frac{2d}{m} -C \right) + C \right],
  \end{align*}
  where
   $ C = \frac{2L^2 d}{\kappa }[h(\kappa^{-1} + h)] \left(2+ \frac{L^2h}{m} + \frac{L^2h^2}{6} \right).$
\end{proof}

\section{Proofs for Section \ref{subsec:tmrmld}}
\label{app:tmulaem}

\begin{proof}[Proof of Theorem \ref{prop:tmuladisc}]
    Consider the $i$--th component of $Z_{k+1}$ and its Taylor expansion around the value $Z_k^{(i)} = T_i(X_k)$. Ignoring terms of order $h^{3/2}$ and higher, we have
    \begin{align*}
        Z^{(i)}_{k+1} =& T_i\left( X_k + h \nabla_x \log \eta(X_k) + \sqrt{2h} \xi_{k+1}\right) \\
         = & T_i(X_k) + (\nabla_x T_i(X_k) )^\top\left(h \nabla_x \log \eta(X_k) + \sqrt{2h} \xi_{k+1} \right)\\& + \frac{1}{2}\left(h \nabla_x \log \eta(X_k) + \sqrt{2h} \xi_{k+1} \right)^\top \nabla_x^2 T_i \left(h \nabla_x \log \eta(X_k) + \sqrt{2h} \xi_{k+1} \right) \\ &+ \mathcal{O}(h^{3/2}) \\
         =& Z_k ^{(i)} + h(\nabla_x T_i)^\top \left(\mathbf{J}_T^\top \nabla_y \log \pi(T(X_k)) + \mathbf{J}_T^\top\nabla_y \log \det \mathbf{J}_T\right) + \sqrt{2h} (\mathbf{J}_T \xi_{k+1})^{(i)} \\ & + h {\xi_{k+1}}^\top \nabla_x^2 T_i \xi_{k+1} + \mathcal{O}(h^{3/2})\\
         =& Z_k^{(i)} + h( \mathbf{J}_T\mathbf{J}_T^\top\nabla_y\log \pi(Z_k) )^{(i)}+ h(\mathbf{J}_T\mathbf{J}_T^\top \nabla_y \log \det \mathbf{J}_T)^{(i)} + h\sum_{j = 1}^d \frac{\partial^2 T_i}{\partial x_j^2} \\ &+  \sqrt{2h} (\mathbf{J}_T \xi^m)_k+hN_k^{(i)} + \mathcal{O}(h^{3/2}),
    \end{align*}
  where $N_k^{(i)} = \xi_{k+1}^\top \nabla_x^2 T_i\xi_{k+1} -  \sum_{j = 1}^d \frac{\partial^2T_i}{\partial x_j^2} $. Note that $N_k^{(i)}$ is a mean zero non-Gaussian random variable. (In fact, it is a sum of scaled and centered $\chi^2$ random variables). Recall that $\xi_{k+1} \sim\mathcal{N}(0,\mathbf{I}_d)$, so $\Ex[\xi_{k+1}^{(j)} \xi_{k+1}^{(l)}] = \delta_{jl}$ and
  \begin{align*}
    \Ex\left[\xi_{k+1}^\top \nabla_x^2 T_i\xi_{k+1} \right]-  \sum_{j = 1}^d \frac{\partial^2T_i}{\partial x_j^2}   =  \sum_{j=1}^d \sum_{l = 1}^d \Ex\left[\xi_{k+1}^{(j)}\xi_{k+1}^{(l)} \right] \frac{\partial^2 T_i}{\partial x_j\partial x_l} - \sum_{j = 1}^d \frac{\partial^2T_i}{\partial x_j^2} =0.
  \end{align*}
    In the proof of Theorem \ref{prop:TMRMLD}, we showed
    %\begin{align*}
        $(\nabla_y\cdot\mathbf{J}_T\mathbf{J}_T^\top)^{(i)} = \sum_{j = 1}^d \frac{\partial^2 T_i}{\partial x_j^2} + \mathbf{J}_T\mathbf{J}_T^\top\nabla_y \log \det \mathbf{J}_T,$
    %\end{align*}
and therefore,
    \begin{align*}
        Z_{k+1}^{(i)} = Z_k^{(i)} + h( \mathbf{J}_T\mathbf{J}_T^\top\nabla_y\log \pi(Z_k) )^{(i)}+ h(\nabla_y \cdot \mathbf{J}_T\mathbf{J}_T^\top)^{(i)}&+  \sqrt{2h} (\mathbf{J}_T \xi_{k+1})^{(i)} \\ &+ N_k^{(i)} + \mathcal{O}(h^{3/2}).
    \end{align*}
    After dropping the non-Gaussian and higher order terms, and substituting $\mathbf{J}_S^{-1}$ for $\mathbf{J}_T$, we arrive at the desired result.

As for the mean-squared error,  note that $\Ex\left[\left\|N_k\right\|^2\right] = \sum_{i = 1}^d \Ex\left[\left|N_k^{(i)}\right|^2\right]$, so we compute
\begin{align*}
    \Ex\left[\left|N_k^{(i)}\right|^2 \right] &= \Ex\left[\hspace{-3pt}\left(\xi_{k+1}^\top \nabla^2 T_i \xi_{k+1}  -  \sum_{j = 1}^d \frac{\partial^2 T_i}{\partial x_j^2}\right)^2 \right]\\ &= \Ex\left[\hspace{-3pt}\left(\sum_{jl} \frac{\partial^2T_i}{\partial x_j\partial x_l}\xi_{k+1}^{(j)}\xi_{k+1}^{(l)} - \sum_{j = 1}^d \frac{\partial^2 T_i}{\partial x_j^2}\right)^2 \right].\nonumber
\end{align*}
Expanding the expression in the expectation, we have
\begin{align*}
    &\left( \sum_{jl} \frac{\partial^2T_i}{\partial x_j\partial x_l}\xi_{k+1}^{(j)}\xi_{k+1}^{(l)}\right)^2 \hspace{-5pt} + \left(\sum_{j = 1}^d \frac{\partial^2 T_i}{\partial x_j^2} \right)^2 \hspace{-5pt}- 2\left( \sum_{jl} \frac{\partial^2T_i}{\partial x_j\partial x_l}\xi_{k+1}^{(j)}\xi_{k+1}^{(l)}\right)\left(\sum_{i = 1}^d \frac{\partial^2 T_i}{\partial x_i^2}\right) \\
    = & \sum_{jlmn} \frac{\partial^2T_i}{\partial x_j\partial x_l} \frac{\partial^2T_i}{\partial x_m\partial x_n}\xi_{k+1}^{(j)}\xi_{k+1}^{(l)}\xi_{k+1}^{(m)}\xi_{k+1}^{(n)} + \sum_{jl} \frac{\partial^2 T_i}{\partial x_j^2}\frac{\partial^2 T_i}{\partial x_l^2} \\ & - 2\sum_{jlm} \frac{\partial^2T_i}{\partial x_j\partial x_l}\frac{\partial^2 T_i}{\partial x_m^2} \xi_{k+1}^{(j)}\xi_{k+1}^{(l)}.
\end{align*}
The expectation of the summands of the first term are only nonzero if pairs of indices are equal or if all the indices are equal. Using the fact that $\Ex\left[\hspace{-2pt}\left(\xi_{k+1}^{(j)} \right)^4\hspace{-1pt} \right] =3$ and $\Ex\left[\hspace{-1pt}\left(\xi_{k+1}^{(j)} \right)^2\hspace{-3pt}\left(\xi_{k+1}^{(l)}\right)^2 \hspace{-1pt}\right] = 1$,
\begin{align*}
  \Ex&\left[\sum_{jlmn} \frac{\partial^2T_i}{\partial x_j\partial x_l} \frac{\partial^2T_i}{\partial x_m\partial x_n}\xi_{k+1}^{(j)}\xi_{k+1}^{(l)}\xi_{k+1}^{(m)}\xi_{k+1}^{(n)} \right]\\ &= 3\sum_{j}\hspace{-3pt} \left(\frac{\partial^2 T_i}{\partial x_j^2} \right)^2 \hspace{-6pt} + \sum_{\substack{jl\\j\neq l}} \frac{\partial^2 T_i}{\partial x_j^2} \frac{\partial^2T_i}{\partial x_l^2} + \sum_{\substack{jl\\j\neq l}}  \left( \frac{\partial^2 T_i}{\partial x_j\partial x_l}\right)^2\hspace{-3pt} \\
  & = \sum_{j}\hspace{-3pt} \left(\frac{\partial^2 T_i}{\partial x_j^2} \right)^2 \hspace{-6pt} + \sum_{\substack{jl}}  \left[\frac{\partial^2 T_i}{\partial x_j^2} \frac{\partial^2T_i}{\partial x_l^2} +  \left( \frac{\partial^2 T_i}{\partial x_j\partial x_l}\right)^2 \right]
\end{align*}
For the third term, we have
%\begin{align*}
 $ \Ex\left[- 2\sum_{jlm} \frac{\partial^2T_i}{\partial x_j\partial x_l}\frac{\partial^2 T_i}{\partial x_m^2} \xi_{k+1}^{(j)}\xi_{k+1}^{(l)}\right] = - 2\sum_{j,l} \frac{\partial^2 T_i}{\partial x_j^2} \frac{\partial^2 T_i}{\partial x_l^2}.$
 %\end{align*}
Summing terms together, we obtain
\begin{align*}
  \Ex\left[\left|N_k^{(i)}\right|^2 \right] = \sum_{j = 1}^d \left(\frac{\partial^2 T_i}{\partial x_j^2} \right)^2 + \sum_{j = 1}^d \sum_{l = 1}^d \left( \frac{\partial^2 T_i}{\partial x_j\partial x_l}\right)^2.
\end{align*}
To obtain the desired result, we sum over the above expression for all $i = 1,\ldots d$.

\end{proof}

\begin{proof}[Proof of Lemma \ref{L:LogSobolev}]
The one-dimensional statement of this lemma is given without details in Proposition 5.4.3 of \cite{BakryBook}. Here we provide for completeness the details for the general multidimensional case. Since $\eta$ satisfies $\text{LSI}(\alpha)$, a change of variables leads to
\begin{align}
\textrm{Ent}_{\pi}[f]&=\textrm{Ent}_{\eta}[f\circ T]\leq \frac{1}{2\alpha}\int_{\mathbb{R}^{d}}\frac{\|\nabla \left(f\circ T\right)(x)\|^{2}}{\left(f\circ T\right)(x)}\eta(x)dx\nonumber\\
%&= \frac{1}{2\alpha}\int_{\mathbb{R}^{d}}\frac{\|J^{\top}_{T}(x)\nabla f(T(x))\|^{2}}{\left(f\circ T\right)(x)}\eta(x)dx\nonumber\\
&= \frac{1}{2\alpha}\int_{\mathbb{R}^{d}}\frac{\|\J^{\top}_{T}(x)\nabla  f(T(x))\|^{2}}{f(T(x))}\pi(T(x))\det \J_{T}(x)dx\nonumber\\
&= \frac{1}{2\alpha}\int_{T^{-1}(\mathbb{R}^{d})}\frac{\|\J^{\top}_{T}(T^{-1}(y))\nabla f(y)\|^{2}}{ f(y)}\pi(y)\left[\det \J_{T}(T^{-1}(y)) \det \J_{T^{-1}}(y)\right]dy.\nonumber
\end{align}

Since, we shall have that  $\det \J_{T}(T^{-1}(y)) \det \J_{T^{-1}}(y)=1$ and due to $T^{-1}(\mathbb{R}^{d})=\mathbb{R}^{d}$:
\begin{align}
\textrm{Ent}_{\pi}[f]
&\leq \frac{1}{2\alpha}\int_{\mathbb{R}^{d}}\frac{\|\J^{\top}_{T}(T^{-1}(y))\nabla f(y)\|^{2}}{ f(y)}\pi(y)dy.\nonumber
\end{align}

Notice that the Cauchy-Schwarz inequality gives
\begin{align}
\|\J^{\top}_{T}(T^{-1}(y))\nabla f(y)\|^{2}&=\sum_{i=1}^{d}\left|\sum_{k=1}^{d}\frac{\partial T_{k}}{\partial x_{i}}(x)\Big{|}_{x=T^{-1}(y)} \partial_{y_{k}}f(y)\right|^{2} \\
%&\leq \sum_{i=1}^{d}\sum_{k=1}^{d}\left|\frac{\partial T_{k}}{\partial x_{i}}(x)\Big{|}_{x=T^{-1}(y)} \right|^{2} \sum_{k=1}^{d}\left|\partial_{y_{k}}f(y)\right|^{2}\nonumber\\
&\leq \|\J^{\top}_{T}(T^{-1}(y))\|_{2}^{2} \left\|\nabla f(y)\right\|^{2}. \nonumber\end{align}

Since $\sup_{y\in\mathbb{R}^{d}}\|\J^{\top}_{T}(T^{-1}(y))\|_{2}^{2}\leq \frac{1}{\rho^{2}}$, we shall have
\begin{align}
\textrm{Ent}_{\pi}[f]
\leq \frac{1}{2\alpha\rho^{2}}\int_{\mathbb{R}^{d}}\frac{\left\|\nabla f(y)\right\|^{2}}{ f(y)}\pi(y)dy
\nonumber
\end{align}
which means that $\pi$ satisfies $\text{LSI}(\alpha \rho^{2})$, completing the proof of the lemma.
\end{proof}

\medskip

\end{document}